\documentclass[11pt,nofootinbib,superscriptaddress]{revtex4-1}

\usepackage[english]{babel}
\usepackage[T1]{fontenc}
\usepackage[latin9]{inputenc}
\usepackage{geometry}
\usepackage{amsmath,amssymb,bbm,amsthm,MnSymbol}
\usepackage{graphicx}
\usepackage{color}
\usepackage{esint}

\usepackage{hyperref}

\newcommand{\defref}[1]{definition \ref{#1}}

\newtheorem{lemma}{Lemma}
\newtheorem{theorem}{Theorem}
\newtheorem{definition}{Definition}

\newtheorem{corollary}{Corollary}

\def\beq{\begin{equation}}
\def\eeq{\end{equation}}
\def\bq{\begin{equation*}}
\def\eq{\end{equation*}}

\def\N{\mathbb{N}}

\def\R{\mathbb{R}}
\def\C{\mathbb{C}}

\def\dif{\mathrm{d}}

\newcommand{\bessel}[2]{\mathrm{J}_{#1}\!\left(#2\right)}
\newcommand{\mbessel}[2]{\mathrm{I}_{#1}\!\left(#2\right)}
\newcommand{\mbesselt}[2]{\mathrm{K}_{#1}\!\left(#2\right)}
\newcommand{\laguerre}[3]{\mathrm{L}_{#1}^{#2}\!\left(#3\right)}

\def\hilbert{\mathcal{H}}

\newcommand{\cfunction}[2][\infty]{\mathrm{C}^{#1}(#2)}
	
\newcommand{\ket}[1]{|#1\rangle}
\newcommand{\bra}[1]{\langle #1|}
\newcommand{\scal}[1]{\langle #1\rangle}

\def\config{\mathcal{C}}
\def\phasespace{\mathcal{M}}
\def\observ{\mathcal{O}}
\def\hamVF{\chi}
\def\Lie{\mathcal{L}}

\newcommand{\cotb}[1]{\mathrm{T}^{*}#1}
\newcommand{\compl}[1]{#1^{\C}}

\def\su{\mathfrak{su}}

\renewcommand{\exp}[1]{\operatorname{exp}\left(#1\right)}
\newcommand{\ee}[1]{\operatorname{e}^{#1}}

\renewcommand{\Re}{\mathrm{Re}}
\renewcommand{\Im}{\mathrm{Im}}

\def\phys{\text{phys}}

\def\m{\mathfrak{m}} 

\def\w{\mathrm{w}}

	\def\SFreview{\cite{Baez:1999sr,Perez:2012wv,Perez:2003vx}}
	
	\def\LQCreview{\cite{Ashtekar:2008zu,Bojowald:2011ku, Ashtekar:2003hd, Ashtekar:2011ni}}
	\def\arnold{\cite{arnold}}
	
	\def\CCS{\cite{Thiemann:2002vj, Sahlmann:2001nv,Thiemann:2000bw, Thiemann:2000ca,Thiemann:2000bx}}

 \date{\normalsize \today} 

\begin{document}

\author{Antonia Zipfel}
\email{antonia.zipfel@fuw.edu.pl}
\affiliation{Instytut Fizyki Teoretycznej, Uniwersytet Warszawski, Pasteura 5, 02-093 Warszawa, Poland, EU}

\author{Thomas Thiemann}
\email{thomas.thiemann@gravity.fau.de}
\affiliation{Universit\"at Erlangen, Institut f\"ur Quantengravitation, Staudtstrasse 7, D-91058 Erlangen, EU}

\title{\bf Stable coherent states}

\begin{abstract}
\begin{center}
{\bf Abstract}
We analyze the stability under time evolution of complexifier coherent states (CCS) in one-dimensional mechanical systems. A system of coherent states is called stable if it evolves into another coherent state. It turns out that a system can only poses stable CCS if the classical evolution of the variable $\mathbf{z}=\ee{-i\Lie_{\hamVF_C}} \mathbf{q}$ for a given complexifier $C$ depends only on $\mathbf{z}$ itself and not on its complex conjugate. This condition is very restrictive in general so that only few systems exist that obey this condition. However, it is possible to access a wider class of models that in principle may allow for stable coherent states associated to certain regions in the phase space by introducing action-angle coordinates.
\end{center}
\end{abstract}
\maketitle

\section{Introduction}

Coherent states have proven to be a powerful tool in many areas of physics as well as mathematics. The name `coherent' goes back to Glauber \cite{Glauber1963_1,Glauber1963_2} who rediscovered Schr\"odinger's states in the context of quantum optics. They are also used, for example, in geometric quantization \cite{Woodhouse:gq1991}, harmonic analysis and representation theory \cite{BargmannIntTrafo,Segal:mp1960, Hall:ha2000}. This broad applicability entailed a vast number of generalizations, just to mention some \cite{PerelomovGCS,Hall:sb1994, KlauderSkag, Klauder:2001ra}.

In quantum gravity coherent states are employed to derive a semiclassical limit of the model in question. Especially in the absence of experimental data, this can provide important insights on quantization ambiguities and possible inconsistencies. In canonical loop quantum gravity (LQG) one uses, for example, so-called complexifier coherent states {\CCS}, going back to the pioneering work of Hall \cite{Hall:sb1994}, in order to define such a limit. For constraint systems such as gravity one has to decide on which space, the kinematical or the physical Hilbert space, the states shall be defined. Which strategy is chosen depends, of course, on the problem in question but in many cases it is easier to build coherent states on the kinematical rather than the physical Hilbert space. This truly applies to LQG where physical states are only known formally and brings in a new aspect that has to be respected as the states designed on the kinematical Hilbert space should not lose their properties when solving the constraints.

From an heuristic point of view the implementation of a constraint $\hat{H}$ is related to a sort of time-evolution generated by $\hat{H}$ since 
\beq
\label{eqn:stab1}
\psi_{\phys} `=' \delta(\hat{H})\,\psi `=' \int\dif t\, \ee{it\hat{H}} \psi
\eeq
gives a formal solution. In fact, many strategies such as group averaging and rigging map procedures (see e.g. \cite{ThiemannBook}) take this as a starting point. Evidently, this ansatz is also advocating itself in order to solve the Hamiltonian constraint of LQG and is the initial idea from which spin foam {\SFreview} models arose. So instead of asking "Is a coherent state maintaining its coherence when solving the constraints?" it is tempting to simplify matters and ask: "Is the coherent state $\psi_z$ stable under the evolution generated by $\hat{U}(t):=\ee{it\hat{H}/\hbar}\,$?" or likewise "Is $\hat{U}(t)\psi_z$ still coherent?" These questions are as well of interest in quantum mechanics because mostly one is not only interested in the semiclassical behavior at a certain time but in the dynamical evolution.

In this work, the necessary conditions for the existence of stable complexifier coherent states are investigated. It is found that in general it is very hard to construct a complexifier adapted to the dynamics of a given model. Nevertheless, the derived criteria are form invariant under canonical transformation which opens the possibility to excess a wider class of models, namely, those that show a quasi-periodic motion.

In the following section, the semiclassical properties of coherent states (section \ref{sec:Semiclassical}) and the construction principle of complexifier coherent states (section \ref{sec:CCS}) are reviewed based on \cite{ThiemannBook}. Thereafter, a stability criterium for finite dimensional models will be derived and in section \ref{sec:first_solution} we will discuss a simplified ansatz to find solutions to this condition. Section \ref{no-go} contains a proof that it is, in fact, not possible to use this simplified ansatz to determine other systems than the harmonic and radial oscillator that posses stable complexifier states. A generalized construction principle using so-called action-angle coordinates and the Hamilton-Jacobi approach is given in section \ref{sec:Generalized_construction} and some examples are analyzed in section \ref{sec:examples}. The paper closes with a short discussion of the results.

\section{Coherent states}

\subsection{Preliminaries and conventions}
\label{sec:partIconv}

If not stated otherwise it will be assumed that the phase space $\mathcal{M}$ of a given system with finite number of degrees of freedom $\mathfrak{f}$ is the cotangent bundle $\cotb{\config}$ of the configuration space $\config$. The Hamiltonian vector field $\hamVF_f$ of a continuous differentiable function $f$ on $\phasespace$ is the vector field that satisfies the condition $0\equiv\Lie_{\hamVF_f} \Omega$ where $\Omega$ is the symplectic 2-form on $\phasespace$ and $\Lie_{\hamVF_f}$ the Lie-derivative along $\hamVF_f$. The Poisson bracket corresponding to $\Omega$ is given by
\bq
\{f,g\}:=\Omega(\hamVF_f, \hamVF_{g})=\hamVF_f[g]
\eq
and multiple Poisson brackets are defined through the recursion relation $\{f,g\}_{(n+1)}:=\{f,\{f,g\}_{(n)}\}$ with $\{f,g\}_{(0)}:=g$. The Liouville measure is the measure on $\mathcal{M}$ which is invariant under the action of the symplectic group that preserves $\Omega$.  

Throughout the rest of this paper, $\mathbf{x}$ will denote the $\mathfrak{f}$-tuple $(x_1,\cdots, x_{\mathfrak{f}})$, $\mathbf{x}\cdot\mathbf{y}=\sum_{j=1}^{\mathfrak{f}} x_j\,y_j$ the usual Euclidean scalar product and $(\mathbf{p},\mathbf{q})$ a canonical conjugated pair, that is, they satisfy 
\bq
\{p_j,q_k\}=\delta_{jk}
\quad\text{and}\quad
\{p_j,p_k\}=0=\{q_j,q_k\}~.
\eq
Furthermore, $\mathcal{O}$ will denote a sub-algebra of the Poisson-algebra $\cfunction{\config}$ that separates the points of $\phasespace$ and $\overline{z}$ the complex conjugate. 

Under quantization we understand a map $(\phasespace, \{\cdot,\cdot\},\mathcal{O})\to (\hilbert, \frac{1}{i\hbar}[\cdot,\cdot],\widehat{\mathcal{O}})$ where $\hilbert$ is a Hilbert space and $\widehat{\mathcal{O}}$ is a sub-algebra of the algebra of linear operators $\mathcal{L}(\hilbert)$ on $\hilbert$ that is a representation of $\mathcal{O}$. If not said otherwise the Hilbert space $\hilbert$ is the space of square integrable functions $L^2(\overline{\config},\dif \mu)$ on a suitable extension $\overline{\config}$ of the configuration space with measure $\dif\mu$. The scalar product on $\hilbert$ is usually given by
\bq
\scal{f,g}=\int\dif\mu(x)\, \overline{f(x)}\,g(x)~.
\eq

\subsection{Semiclassical and coherent states}
\label{sec:Semiclassical}
In his lecture "\"Uber die Spektraltheorie der Elemente" \cite{Bohr1920}, given in 1920 at a meeting of the german physical society in Berlin, N. Bohr introduced the principle that the behavior of a quantum system should mimic the classical one for high energies. To formulate this statement in a more precise manner it is useful to introduce the notion of \emph{semiclassical states}. This are elements $ \psi_{\m}$ in the Hilbert space $\hilbert$ that are associated to points $\m$ in $\phasespace$. They are constructed in such a way that the expectation value of a quantum observable $\hat{O}$ in a given subalgebra $\widehat{\observ}\subset\mathcal{L}(\hilbert)$ that separates the points of $\mathcal{M}$ is close to the classical value $O(\m)$ of the corresponding phase space function $O$. Stated differently, the following properties have to hold for all generic points $\m\in\phasespace$ (i.e. points for which the denominator in \eqref{eqn:sc1},\eqref{eqn:sc2} and \eqref{eqn:sc3} is non zero):
\begin{itemize}
\item\textit{Expectation value property}
\beq
\label{eqn:sc1}
\left|\frac{\scal{\psi_{\m},\hat{O}\,\psi_{\m}}}{O(\m)}-1\right|\ll 1
\eeq
\item\textit{Ehrenfest property}
\beq
\label{eqn:sc2}
\left|\frac{\scal{\psi_{\m},[\hat{O},\hat{O}']\,\psi_{\m}}}{i\hbar\,\{O,O'\}}-1\right|\ll1
\eeq
\item\textit{Fluctuation property}
\beq
\label{eqn:sc3}
\left|\frac{\scal{\psi_{\m},\hat{O}^2\,\psi_{\m}}}{\scal{\psi_{\m},\hat{O}\,\psi_{\m}}^2}-1\right|\ll 1
\eeq
\end{itemize}
For most systems it is not possible to design semiclassical states for \emph{all} observables simultaneously but it highly depends on the chosen subalgebra. A good example are the coherent states of the harmonic oscillator introduced by Schr\"odinger in 1926 \cite{Schroedinger1926} which have `good' semiclassical properties for the linear span of the annihilator $\hat{a}$, the creator $\hat{a}^{\dagger}$ and $\mathbbm{1}$. In addition, these states have several other desirable properties which motivates the following definition.
\begin{definition}[Coherent states]
\label{def:CS}
A system of states $\{\psi_{\m}\}_{\m\in\phasespace}\subset\hilbert$ is said to be coherent provided that in addition to \eqref{eqn:sc1},\eqref{eqn:sc2} and \eqref{eqn:sc3} the states also obey:
\begin{itemize}
\item {\bf Overcompleteness (Resolution of identity):}
\beq
\label{eqn:overcomplete}
\mathbbm{1}_{\hilbert}=\int_{\phasespace} \dif\nu(\m)\, \ket{\psi_{\m}}\bra{\psi_{\m}}
\eeq
for some measure $\nu$ on $\phasespace$.

\item {\bf Annihilation operator property:} There exist operators $\hat{z}$ such that $\hat{z}\,\psi_{\m}=z(\m)\,\psi_{\m}$~.

\item {\bf Minimal uncertainty:} For the self-adjoint operators $\hat{x}=(\hat{z}+\hat{z}^{\dagger})/2$ and $\hat{y}=(\hat{z}-\hat{z}^{\dagger})/(2i)$ the Heisenberg uncertainty relation is saturated, i.e.
\beq
\label{eqn:minUncert}
\scal{(\hat{x}-\scal{\hat{x}^2}_{\m})^2}_{\m}=\scal{(\hat{y}-\scal{\hat{y}^2}_{\m})^2}_{\m}
=\frac{\hbar}{2}|\scal{[\hat{x},\hat{y}]}_{\m}|
\eeq
where $\scal{\cdot}_{\m}:=\scal{\psi_{\m},\cdot\,\psi_{\m}}$.

\item{\bf Peakedness property:}
For any $\m\in\phasespace$, the overlap function
\beq
\m'\mapsto |\scal{\psi_{\m},\psi_{\m'}}|^2
\eeq
is concentrated in a phase space cell of Liouville volume $\frac{1}{2}|\scal{[\hat{p},\hat{q}]}_{\m}|$.
\end{itemize}
\end{definition}
As stated in the introduction, coherent states have a broad application in many areas of physics and mathematics, which entailed a vast number of generalizations, so that by now `coherent state'  is not a clear-cut expression in the literature. In the subsequent section the generalization suggested in {\CCS} is reviewed, which in most cases preserves the properties mentioned in \defref{def:CS}.

\subsection{The complexifier method}
\label{sec:CCS}
The central point of semiclassical/coherent states is that they are \emph{continuously labeled by points in the classical phase space} $\phasespace$. In the case of the harmonic oscillator, this is achieved by constructing the eigenstates of the operator $\hat{a}$ corresponding to the \emph{complex parametrization} 
\beq
\label{HO}
a:=\sqrt{\frac{m\omega}{2}}(q-\frac{i}{m\omega}p)~.
\eeq
It is, of course, just one possible choice out of many different parametrizations of $\phasespace$ and therefore many other state systems are imaginable that are coherent in the sense of definition \ref{def:CS}. This is the starting point of the complexifier method introduced in {\CCS} that enables to directly relate a complex parametrization of $\mathcal{M}$ with a state system in $\hilbert$. 

Let $C:\phasespace\to\R$ be a positive definite function with the dimension of an action that is smooth with respect to the Liouville measure on $\phasespace$, has a nowhere vanishing Hamiltonian vector field $\chi_C$ and grows stronger than linearly in $\mathbf{p}$ for each fixed point $\mathbf{q}\in\config$. A function satisfying these properties is called a \emph{complexifier} since 
\beq
\label{eqn:compl1}
q_i\mapsto z_i:=\ee{-i\Lie_{\hamVF_C}} q_i
\eeq
yields a complex coordinate system on $\phasespace$ for a given parametrization $\{q_i| i=1,\cdots\mathfrak{f}\}$ of the configuration space $\mathcal{C}$. The smoothness of $C$ and the fact that $\hamVF_C$ is nowhere vanishing guarantee that  $C$ induces a non-degenerate, smooth transformation. If $\mathcal{M}=\cotb{\config}$ then $\ee{-i\Lie_{\hamVF_C}}q_i$ defines a symplectomorphism\footnote{This is no longer true for $\phasespace\neq\cotb{\config}$ but $z_i$ still provide good local coordinates due to Darboux's theorem.} that maps a point $\m\in\phasespace$ to a point $z(\m)$ in the complex extension $\compl{\config}$ of the configuration space.
For example, the complexifier $C=\frac{p^2}{2}$ generates the symplectomorphism 
\bq
\phasespace=\R^2\to \C~,\quad (q,p)\mapsto a= q-i\,p~.
\eq
According to Bargman and Segal \cite{Segal:mp1960,BargmannIntTrafo}, the associated coherent states $\psi_{a}$ also give rise to a transformation from $L^2(\R,\dif\mu)$ onto the space $\mathfrak{H}^2(\C,\dif\nu)$ of square integrable, holomorphic functions on $\C$ through the integral transform
\begin{gather}
\begin{gathered}
\label{eqn:S-B_transform}
f(q)\mapsto [B f](\alpha)=\int_{\R}\dif\mu(q)\,\psi_{\alpha}(q)\, f(q)~.
\end{gathered}
\end{gather}
If the measure $\nu$ is fixed by the resolution of the identity \eqref{eqn:overcomplete} then \eqref{eqn:S-B_transform} is even unitary. Thus, the harmonic oscillator states play a similar role as plane waves for the Fourier transformation. In fact, for this specific example with $C=\frac{p^2}{2}$ one can show that up to a complex phase $\psi_a(x)$ is equal to 
\begin{align}
\label{eqn;exampleCS}
\begin{split}
\exp{-\frac{1}{2\hbar}(x-a)^2}
		&\propto\frac{1}{2\pi}\int\dif{k}\ee{-k^2/(2t)}\ee{i k(x-a)}\\
 		&= \left[\ee{-\hat{C}/\hbar}\delta_y(x)\right]_{y\to a}
\end{split}
\end{align}
where $y\to a$ denotes analytic continuation. 
\begin{definition}[Complexifier coherent states] 
A coherent state associated to a complexifier $C\in\cfunction{\phasespace}$ is an element of $\hilbert$ of the form
\beq
\label{eqn:ccs} 
\psi_{\m}(q)=\left[\ee{-C/\hbar}\delta_{q'}(q)\right]_{q'\to z(\m)}
\eeq
where $q'\to z(\m)$ denotes analytic continuation to $z(\m)=[\ee{-i\Lie_{\hamVF_C}} \,q](\m)$.
\end{definition}
This also explains why it was required that $C$ has the dimension of an action and has to satisfy a growth condition. Namely, if $\hat{C}/\hbar$ would not be dimensionless then the exponential would not be well-defined. Apart from that, $\ee{-\hat{C}/\hbar}$ must decay fast enough to smooth out the divergence of $\delta$. In equation \eqref{eqn;exampleCS} this is achieved by the gaussian factor that is added through the action of  $\ee{-\hat{C}/\hbar}$. 

By construction, the states \eqref{eqn:ccs} are eigenstates of the annihilators
\beq
\label{eqn:annihil}
\hat{z}_i=\ee{-\hat{C}/\hbar} \hat{q}_i \ee{\hat{C}/\hbar}
\eeq
resulting from quantizing \eqref{eqn:compl1} and therefore automatically obey  minimal uncertainty \eqref{eqn:minUncert}. Also condition \eqref{eqn:sc1} is clearly satisfied for the subalgebra spanned by $\mathbf{z}$ and $\mathbf{\bar{z}}$. All the other properties required in \defref{def:CS} do not follow directly but are very plausible since the states are essentially a regularization of a distributions. Therefore, it is reasonable to expect that they generate a resolution of identity and are  well-peaked on the phase space.


\section{Stability of coherent states in quantum mechanics}

\subsection{A stability criterion}
\label{sec:Stability}

Recall that the classical evolution is generated by the flow of the Hamiltonian vector field $\hamVF_H$ of the Hamiltonian $H$. For a not explicitly time dependent Hamiltonian this means that a phase space function $f$ evolves as $f(t)=\ee{(t-t_0) \Lie_{\hamVF_H}}f(t_0)$ and the quantum evolution of such a model is described by the operator  $\hat{U}(t,t_0)=\ee{\frac{1}{i\hbar} (t-t_0)\hat{H}}$.
\begin{definition}[Stable Coherent states]
\label{def:stable}
A coherent state $\psi_{z(t_0)}$ is stable under the evolution $\hat{U}(t,t_0)$  iff
\begin{align}
\label{eqn:stab}
\hat{U}(t,t_0)\,\psi_{z(t_0)}=\ee{i\, \lambda(t)}\psi_{z(t)}\quad \forall\; t\in\R^+
\end{align}
where $z(t)$ follows the classical motion of $z(t_0)$ on the phase space. The system of coherent states, $\{\psi_{z(\m)}|\m\in\phasespace\}$, is called stable iff all states are stable. 
\end{definition}
A simple example for a stable state system are the original coherent states $\psi_a$ of the harmonic oscillator. By a short calculation using $\psi_a= \ee{|a|^2/2}\sum_n \frac{a^n}{n!}\,\ket{n}$ and $\hat{H}_{ho}\ket{n}=\hbar\omega \, (n+1/2)\ket{n}$ one finds 
\bq
\hat{U}_{ho}(t)\,\psi_{a}
=\ee{it\omega/2}\psi_{a(t)}~.
\eq
Even more general, the set $\{\psi_{a}\}$ is stable under the evolution generated by a Hamiltonian\footnote{If $H$ is required to be self-adjoint then $f$ must be linear in $\hat{a}$ and $H$ must be of the form $H=\omega(t)a^{\dagger}a+f(t)a^{\dagger}+f^{\ast}(t)a+\beta(t)$. Note, there are even more general Hamiltonians under which certain proper subsets of coherent states are stable (see \cite{Kano:rt1976,Letz:eo1977}).}
that satisfies $\frac{1}{i\hbar}[\hat{H},\hat{a}]=i f(\hat{a},t)$ (see e.g. \cite{Mehta:te1966,Chand:nc1978}). For complexifier coherent states one can proof a very similar criterion: 
\begin{theorem}
\label{th:stabCCS}
Suppose the set of complexifier states $\mathcal{S}_{t_0}:=\{\psi_{z(\m(t_0))}|\m(t_0)\in\phasespace\}$ is over-complete and stable and the time evolution $\hat{U}(t,t_0)$ is unitary then 
\begin{align}
\label{eqn:condStab}
\frac{\dif}{\dif t}\,\hat{z}_j(t_0)=i f_j(\hat{z}_1,\dots \hat{z}_{\mathfrak{f}},t_0) \quad \forall\; j=1,\cdots, \mathfrak{f}~.
\end{align}
On the other hand, if equation \eqref{eqn:condStab} holds at $t_0$ then there exist an $\epsilon>0$ such that  $\mathcal{S}_{t}$ is stable with respect to $\hat{U}(t,t_0)$ for all $|t_0-t|<\epsilon$. 
\end{theorem}
\begin{proof} \mbox{}\\
According to definition \ref{def:stable} the system ${\cal S}_{t_0}$ is stable iff
\begin{align*}
&\hat{z}_j(t_0) \hat{U}(t,t_0)\ket{\psi_{z(t_0)}}=z_j(t) \hat{U}(t,t_0)\ket{\psi_{z(t_0)}}\\
\Leftrightarrow\quad
&\hat{z}_j(t) \ket{\psi_{z(t_0)}}=z_j(t) \ket{\psi_{z(t_0)}}\\
\Rightarrow\quad
&[\frac{\dif}{\dif t}\hat{z}_j(t_0)] \ket{\psi_{z(t_0)}}=[\frac{\dif}{\dif t}z_j(t_0)]\ket{\psi_{z(t_0)}}~.
\end{align*}
Here, $\hat{z}_j(t):=\hat{U}^{\dagger}(t,t_0) \hat{z}_j(t_0)\hat{U}(t,t_0)$ is the time-dependent operator in the Heisenberg picture. Since $\mathcal{S}_{t_0}$ is over-complete the last equation can only hold if $\frac{\dif}{\dif t}\hat{z}_j(t_0)$ is a function of the annihilators $\hat{z}_j(t_0)$, that is, if $\frac{\dif}{\dif t}\,\hat{z}_j(t_0)=i\, f_j(\hat{z}_1,\dots \hat{z}_{\mathfrak{f}},t_0)$ where $i$ is just introduced for convenience. 

On the other hand, if equation \eqref{eqn:condStab} holds then in some region $|t_0-t|<\epsilon$ the Taylor expansion of $\hat{z}(t)$ exists and reads  
\begin{align*}
\hat{z}_j(t)=\sum_n \frac{t^n}{n!} \frac{\dif^n}{\dif t^n} \hat{z}_j(t_0)=\hat{z}_j(t_0)+i\,(t-t_0)\,f_j(\hat{z}_1(t_0),\dots,\hat{z}_{\mathfrak{f}}(t_0),t_0)+\cdots ~.
\end{align*}
Therefore $\hat{z}_j(t) \ket{\psi_{z(t_0)}}=z_j(t) \ket{\psi_{z(t_0)}}$ in that region. 
\end{proof}
This implies the following two conditions that are necessary for coherent sates to be stable:
 \begin{corollary}
\label{corollary1_Part1}
Stable coherent states can only exist if it is possible to find an complexifier $C$ such that $\mathbf{z}=\ee{-i\Lie_{\hamVF_C}} \mathbf{q}$ obeys
 \begin{align}
\label{eqn:condStab2}
\frac{\dif}{\dif t}\,z_j(t)=i\,f_j(z_1,\dots z_{\mathfrak{f}},t_0) \quad \forall\; j=1,\cdots, \mathfrak{f}~.
\end{align}
\end{corollary}
\begin{corollary}
If neither the complexifier $C$ nor the Hamiltonian $H$ are explicitly time-dependent then $\ee{i\mathcal{L}_{\chi_{C}}}H$ must satisfy
\beq 
\label{eqn:classco}
\ee{i\mathcal{L}_{\chi_{C}}}H=i\,\sum_{j=1}^{\mathfrak{f}} p^j f_j(q_1,\dots,q_{\mathfrak{f}})+g(q_1,\dots,q_{\mathfrak{f}})
\eeq
for some functions $f_j$ and $g$ on $\config$ in order that the complexifier coherent state are stable. 
\end{corollary}
\begin{proof}
Since $\frac{\dif}{\dif t} z_j(t)= \{H,z_j\}$ and since $\ee{-i\Lie_{\hamVF_C}}$ is a symplectomorphism it follows from corollary \ref{corollary1_Part1} that 
\bq 
0=\{H,z_j\}-i\,f_j(z_1(t),\dots,z_{\mathfrak{f}}(t)) =\ee{-i\Lie_{\hamVF_C}}\left[\{\ee{i\Lie_{\hamVF_C}}H,q_j\}-i\,f_j(q_,\dots,q_{\mathfrak{f}})\right]~.
\eq
Yet, $\{\ee{i\Lie_{\hamVF_C}}H,q_j\}$ is equal to $f_j(q_,\dots,q_{\mathfrak{f}})$ if and only if equation \eqref{eqn:classco} holds.
\end{proof}
To summarize: The existence of stable states is closely tied to the classical behavior of the model. The subtlety hereby is that conditions \eqref{eqn:condStab2} and \eqref{eqn:classco} are local. Even though it might be possible to find a complexifier satisfying these conditions in a neighborhood of a point $(p_0,q_0)\in\phasespace$ it does not ensure that the resulting parametrization $\mathbf{z}=\ee{-i\Lie_{\hamVF_C}} \mathbf{q}$ is sensible. For example, $\mathbf{z}$ might not be defined everywhere on the phase space or it might be multivalued. Yet, for the system ${\cal S}$ to be over-complete it is important that $\mathbf{z}$ provides a `good' parametrization:
\begin{lemma}
\label{newlemma}
If the system $S$ is over complete then $\{z_j\}$ is a proper parametrization of $\phasespace$ in the sense that almost every point $\mathfrak{m}\in\mathcal{M}$ is uniquely determined by a $\mathfrak{f}$-tuple $(z_1,\cdots,z_{\mathfrak{f}})$. Here, `almost' means up to a set of measure zero with respect to the measure \eqref{eqn:overcomplete}. 
\end{lemma}
\begin{proof}
If this were not the case then there would exist a subset $U$ in $\phasespace$ with a parametrization $\mathbf{w}$ such that states $\psi_{\w(\mathfrak{m})}$ associated to $\mathfrak{m}\in U$ cannot be expanded in terms of states in $\mathcal{S}$. But then 
\bq
\int_{\phasespace-U} \dif\nu(\mathfrak{m})\; \ket{\psi_{z(\mathfrak{m})}}\bra{\psi_{z(\mathfrak{m})}}\neq\mathbbm{1}_{\hilbert}
\eq
which contradicts the assumptions.
\end{proof}
Note, this does not exclude parametrizations that break down at single points nor does it exclude parametrization that differ in  region A from this in region B if the resolution can be modified accordingly, that is, if $\int_{\phasespace}$ splits into $\int_A +\int_B$. Especially the last point is of interest since in many models the dynamics can be fundamentally different in separated regions and therefore it might be necessary to consider different complexifiers for those regions. It is even possible that in certain areas of $\mathcal{M}$ one cannot find stable states while for others there exist some. For example, inside a finite potential well the motion is periodic which suggest that there exist stable states while outside the well the motion is unbounded. This issues will play a significant role in section \ref{Adapted_Complexifier}. For now the main aim is to find solutions to \eqref{eqn:classco}.
\subsection{First solutions}
\label{sec:first_solution}
For a one-dimensional not-explicitly time dependent system equation \eqref{eqn:classco} reduces to
\begin{equation}
	\label{eq:stab1}
	 	\sum_{n=0}^{\infty}\frac{i^{n}}{n!}\{C,H\}_{\left(n\right)}
		=i\, p\, f(q)+g(q)~,
\end{equation}
where $\{C,H\}_{\left(n\right)}=\{C,\{C,H\}_{\left(n-1\right)}\}$ and $f$ and $g$ are arbitrary smooth functions. To find a first solution of this equation consider the following ansatz: 
\begin{gather}
\label{eqn:ansatz1}
H=\frac{p^2}{2\,m}+V(q), \quad C=\frac{p^2}{2\,m\, \omega} +U(q)\\
\label{eqn:ansatz2}
\text{and}\quad \{C,H\}_{(2N+1)}=0\quad \text{for some } N\geq1
\end{gather}
Here, $m$ is the mass, $V$ and $U$ are twice continuously differentiable real valued functions and $\omega>0$ is a free parameter. Note that the even (odd) Poisson brackets in \eqref{eqn:ansatz1} can only contain summands of even (odd) powers of the momentum. Therefore, $f$ and $g$ in \eqref{eqn:ansatz1} have to be real and satisfy
\begin{align}
	\label{eqn:stab1a}
		&\sum_{n=0}^{N}\frac{\left(-1\right)^n}{\left(2n+1\right)!}\{C,H\}_{\left(2n+1\right)}=p\, f(q)\\
	\label{eqn:stab2a}
		&\sum_{n=0}^{N}\frac{\left(-1\right)^n}{\left(2n\right)!}\{C,H\}_{\left(2n\right)} = g(q)~.
\end{align}
A model with $\mathfrak{f}$ degrees of freedom can allow at most $\mathfrak{f}$ independent commuting Hamiltonian vector fields (first integrals of motion) which is why for $\mathfrak{f}=1$ condition \eqref{eqn:ansatz2} is equivalent to
\beq
\label{eqn:ansatz2a}
\{C,H\}_{(2N)}=\sum_{j=0}^{N} \beta_j C^j~.
\eeq
The restriction on the degree of the polynomial is due to the fact that $\{C,H\}_{(2 N)}$ can be at most a polynomial of degree $2N$ in $p$. For $N=1$ equations \eqref{eqn:stab1a}, \eqref{eqn:stab2a} and \eqref{eqn:ansatz2a} are replaced by 
\begin{gather}
\label{eqn:cond1}
\{C,H\}=p\,f(q),
\\
\label{eqn:cond2}
H-\frac{1}{2}\{C,H\}_{(2)}=g(q),
\\
\label{eqn:cond3}
\text{and}\quad\{C,H\}_{(2)}= \beta_1 C+\beta_0~.
\end{gather}
For this choice the first two conditions in \eqref{eqn:ansatz1} are equivalent to 
\bq 
( \omega^{-1} V'-U')/m=f\quad \text{and}\quad \frac{p^2}{2m\omega}\, f'-f\,U'=2(H-g)
\eq 
where $'$ denotes the derivative with respect to $q$. Thus $f'$ must be equal to $\omega$ and $\beta_1$ must equal $2\,\omega$. Inserting this into \eqref{eqn:cond3} yields
\bq
f \,U'=-2\,\omega[U+\beta_2] 
 \eq
which is solved by $f=\omega q+\alpha$ and $U=\frac{\lambda}{\omega} (q+\alpha/\omega)^{-2}+\beta_2$ for some real constants $\lambda$ and $\alpha$. Since $\alpha$ just defines a shift in the configuration variable it can be set to zero without loos of generality. Concluding,
\begin{gather}
\begin{gathered}
\label{radial_oscillator}
H=\frac{p^2}{2m}+\frac{m}{2}\omega^2 q^2 +\lambda q^{-2}
\\
C=\frac{p^2}{2m\omega} + \frac{\lambda}{\omega} q^{-2}
\end{gathered}
\end{gather}
solves \eqref{eqn:ansatz1} and \eqref{eqn:ansatz2} for $N=1$. If $\lambda$ is zero this reduces to the usual Hamiltonian and complexifier of the Harmonic oscillator. The term $\lambda/q^2$ can be interpreted as the angular momentum contribution of a two dimensional oscillator with constant radius which  justifies the name \emph{radial oscillator}. The complexification $z:=\ee{-i\Lie_{\hamVF_C}} q$ can be computed by solving the `equation of motion'
\beq
\label{eqn:a}
\frac{\dif\;}{\dif s} q=\{C,q\}
\eeq
for $q(s):=\ee{s\Lie_{\hamVF_C}} q$ and then extending it analytically ($s\to - i$). Instead of integrating equation \eqref{eqn:a} one can integrate 
\bq
\frac{\dif\;}{\dif s} \left[\frac{1}{2}\left(\frac{\dif q}{\dif s}\right)^2 +\frac{\lambda}{m\omega^2} q^{-2}\right]=0
\quad\implies\quad
\frac{1}{2}\left(\frac{\dif q}{\dif s}\right)^2 +\frac{\lambda}{m\omega^2} q^{-2}=\frac{C}{m\omega}
\eq
since $\frac{\dif\;}{\dif s}\,C=0$ and $\{C,q\}=p/(m\omega)$. This finally gives
\bq
z=q(s=-i)=\left[ \left(q-i\frac{p}{m\omega}\right)^2-2\frac{\lambda}{m\omega^2}\, q^{-2}\right]^{1/2}~.
\eq
Despite that $z$ depends on a square root, one can derive a well-defined quantum operator whose eigenfunctions $\psi_z=[\ee{-\hat{C}/\hbar}\delta_y]_{y\to z}$ can be expressed in terms of modified Bessel functions of the second kind. Furthermore, the obtained states share all desired properties, that is, they form an over-complete set, minimize the uncertainty, have small fluctuations and are of course stable under the dynamics. More details on the quantization of this system can be found in appendix \ref{app:A}.

\subsection{A no-go theorem}
\label{no-go}

After the method applied in the preceding section has been proven so successful it is worth testing whether more solutions to \eqref{eq:stab1} can be found by generalizing the ansatz \eqref{eqn:ansatz1}. Given two real valued, strictly positive, continuously differentiable functions $\alpha$ and $\beta$ on the configuration space let
\begin{gather}
\label{ansatz2a}
\begin{gathered}
H=\alpha(q)\,\frac{p^2}{2}+V(q)
\quad\text{and}\quad
C=\frac{1}{2}[\beta(q)\,p]^2+U(q)\\[5pt]
\text{and suppose} \quad \{C,H\}_{(2N+1)}=0\quad\text{for some } N>1~.
\end{gathered}
\end{gather}
Since $\beta$ should be non-zero in order that the complexifier is well-defined there exist a canonical transformation sending
\bq
p\to p\,\beta(q)
\quad\text{and}\quad q\to \int_0^q \dif x\; [\beta(x)]^{-1}
\eq
so that the first line in \eqref{ansatz2a} may be exchanged through
\beq
\label{ansatz2}
H=\alpha(q)\, p^2+V(q)\quad\text{and}\quad
C=\frac{p^2}{2}+U(q)
\eeq
without loos of generality. To compute the multiple Poisson brackets of \eqref{ansatz2} it turns out to be handy to introduce the operators $\hat{X}=p\,\frac{\partial\;}{\partial q}$ and $\hat{P}=-U' \frac{\partial\;}{\partial p}$ which yields
\begin{align}
\label{help:A}
\{C,H\}_{(n)}=& \hat{X}^n\, H 
+\sum_{\nu=0}^{n-1} \hat{X}^{n-1-\nu}\,\hat{P}\,\hat{X}^{\nu}\,H
+\sum_{\mu+\nu=0}^{n-2} \hat{X}^{n-1-\nu-\mu}\,\hat{P}\,\hat{X}^{\nu}\,\hat{P}\,\hat{X}^{\mu}\,H
+\cdots~.
\end{align}
The general strategy in the subsequent analysis is to reorder the terms by powers in the momentum. The degree in $p$ of the terms $\hat{X}^{\nu_1}\,\hat{P}\,\hat{X}^{\nu_2} \cdots $ only depends on the number of operators $\sharp \hat{O}$  that are applied, i.e.
\bq
\deg_{\,p}\,\hat{X}^{\nu_1}\,\hat{P}\,\hat{X}^{\nu_2}\, \cdots\, V=\sharp \hat{X}-\sharp \hat{P}
\eq
and
\bq
\deg_{\,p}\,\hat{X}^{\nu_1}\,\hat{P}\,\hat{X}^{\nu_2}\, \cdots\, \alpha\, p^2=\sharp \hat{X}-\sharp \hat{P}+2
\eq
where $\deg_{\,p}\,$ denotes the degree in $p$. Note that the total number of operators, $\sharp \hat{X}+\sharp \hat{P}$, corresponds to the grade $n$ of the bracket $\{C,H\}_{(n)}$. For this reason the bracket either contains only even or odd powers in $p$ depending on whether $n$ is even or odd respectively. This can be made more explicit by replacing \eqref{help:A} through
\beq
\{C,H\}_{(n)}=\sum_{\stackrel{\mu\in\N}{0\,\leq\,\mu\,\leq\,(n+2)/2}} \!\!f^{n}_{n+2-2\mu}\; p^{n+2-2\mu}~.
\eeq
The coefficients $f^{n}_{n+2-2\mu}$ are derived from \eqref{help:A}, that is, 
\begin{align}
\label{eqn:f}
\begin{split}
f^{n}_{n+2}:= \alpha^{(n)}~,\quad
f^{n}_{n}:=-\sum_{\nu}^{n-1}(\nu+2) \left(\frac{\partial\;}{\partial q}\right)^{n-1-\nu} U'\;\alpha^{(\nu)}
		+ V^{(n)}~,\quad\cdots
\end{split}
\end{align}
where $g^{(n)}$ is the $n$th derivative of a function $g$ with respect to $q$ and the upper index of $f^n_m$ revers to the total number of operators while the lower one indicates the power in $p$. 
Instead of trying to directly compute these coefficients it is much more useful to consider the following recursion relation: 
\beq
\label{recursion}
f^{n}_{m}= \frac{\partial\;}{\partial q}\, f^{n-1}_{m-1} - U' \, (2m+1)\,  f^{n-1}_{m+1}, \quad m\leq n+2,
\eeq
with $f^0_0\equiv V$ and $f^0_2\equiv\alpha$. 

Furthermore, due to $\frac{1}{2}\,\deg_{\,p} \,\{C,H\}_{(2N)}\leq N+1$ and $\{C,H\}_{(2N+1)}=0$ one finds that 
\beq
\label{eq:power0}
\{C,H\}_{(2N)}=\sum_{m=0}^{N+1} a_{m}\, C^{m}
\eeq
for some constants $a_{m}$. Apart from that, all summands in \eqref{eq:stab1} whose degree in $p$ is exceeding one must cancel which implies that 
\beq
\label{cancel1}
\begin{split}
\sum_{n=m-1}^N \frac{(-)^n}{(2n)!}\, f^{2n}_{2m}=0
\quad\text{and}\quad
\sum_{n=m-1}^{N-1} \frac{(-)^n}{(2n+1)!}\, f^{2n+1}_{2m+1}=0
\end{split}
\eeq
 for all $m>1$. Note, there can be only one non-trivial term of power $2N+2$ and one of $2N+1$ which is why $ f^{2N}_{2N+2}$ and $ f^{2N-1}_{2N+1}$ have to vanish. This in turn implies $\deg_{\,C} \,\{C,H\}_{(2N)}\leq N$ and $\deg_{\,q} \alpha\leq 2N-2$ since $f^{2N-1}_{2N+1}$ is proportional to $\alpha^{(2n-1)}$.
The next non-trivial contributions are those of degree $2N$ and $2N-1$. Due to equation \eqref{cancel1} the  coefficients $f^{2N}_{2N}$ and $f^{2N-1}_{2N-1}$ must satisfy
\beq
\label{help:B}
f^{2N}_{2N}=(2 N)(2N-1)f^{2N-2}_{2N}
\quad\text{and}\quad
f^{2N-1}_{2N-1}=(2N-1)(2N-2)f^{2N-3}_{2N-1}~.
\eeq
Yet, taking the derivative of the equation on the right hand side and inserting equality \eqref{recursion} yields
\bq
\frac{\partial\;}{\partial q}\; f^{2N-1}_{2N-1}=(2N-1)(2N-2) \frac{\partial\;}{\partial q}\;f^{2N-3}_{2N-1}
\quad\Leftrightarrow\quad 
f^{2N}_{2N}=(2N-1)(2N-2)f^{2N-2}_{2N}
\eq
which obviously contradicts the first equation in \eqref{help:B}. Thus, $f^{m}_{2N}$ must vanish for all $m$. With the next term one can proceed in the same manner. Combining \eqref{cancel1}, \eqref{recursion} and $f^{2N+1}_{2N-1}=0$ gives
\bq
0=\frac{\partial\;}{\partial q}\sum_{n=N-2}^N \frac{(-)^{n}}{(2n)!} f^{2n}_{2N-2}
=\sum_{n=N-2}^{N-1} \frac{(-)^{n}}{(2n)!} f^{2n+1}_{2N-1}~.
\eq
This in turn contradicts the second equation in \eqref{help:B} and therefore implies $f^{m}_{2N-1}=0$ for all $m$ which also reduces the degree of the polynomials $\alpha$ and $\{H,C\}_{2N}$ and proves 
\begin{lemma}
\label{lem:vanishing}
For $N>1$ there exist an integer $M$ with $2\leq M<2N-1$ such that $f_m^{n}=0$ for all $m > M$ and all $n$.
\end{lemma}

The above reasoning can be repeated to show that also summands of lower degree in the momentum have to vanish. But note, the number of conditions needed to derive a contradiction is increasing when the degree in $p$ is decreasing since there are more and more non-trivial coefficients that contribute. 
Nevertheless, one can obtain a new condition for the coefficients $f^{2n+1}_{\ast}$ from
\beq
\label{cond1}
0=\sum_{n= m-2}^N \frac{(-)^{n}}{(2n)!}   f^{2n}_{2m}\quad \forall \; m>0
\eeq
by taking the derivative and using \eqref{recursion}. More specifically, this yields
\begin{align}
\nonumber
0&=\sum_{n= m-2}^{N-1} \frac{(-)^{n}}{(2n)!}   \left[f^{2n+1}_{2m+1} -(2m+2)\,U' f^{2n}_{2m+2}\right]\\
\label{cond2}
&=\sum_{n= m-2}^{N-1} \frac{(-)^{n}}{(2n)!}   f^{2n+1}_{2m+1}\quad \forall \; m>0 ~.
\end{align}
From this one can deduce another condition,
\beq
\label{cond3}
\sum_{n= m-2}^{N-1} \frac{(-)^{n}}{(2n)!}   f^{2n+2}_{2m+2}\quad \forall \; \mu>0~,
\eeq
by first taking the derivative and then applying \eqref{recursion} and \eqref{cond2} and so forth. Repeating this procedure also for the odd terms generates $2m-1$ independent conditions for the terms of power $2m$ and $2m+1$. On the other hand, the number of non-trivial terms $\sharp f^{\ast}_m$ in a tower of constant power $m$ increases as $\sharp f^{\ast}_{2m}= N-m+1$ and $\sharp f^{\ast}_{2m+1}= N-m+2$. Since the procedure breaks up as soon as $\sharp f^{\ast}_m$ is greater than the number of available conditions, namely if  $\sharp f^{\ast}_{2m/2m+1}\geq 2m-1$, another trick is needed to eliminate more coefficients.
\begin{lemma}
\label{lem:U}
Let $C=\frac{p^2}{2}+U(q)$, $F=\sum\limits_{n=0}^m F_n\,C^n$ and let $G=G(p,q)$ be a phase space function that is polynomial in $p$. 
\begin{itemize}
\item[(a.)] If $\deg_{\,p} \,G= 2m+1$ then $\{C,G\}=F(C)$ has a solution iff $U=c_1\,q+c_0$ for some constants $c_1,c_0\in\R$.

\item[(b.)] If $\deg_{\,p} \,G= 2m-1$ then $\{C,G\}=F(C)$ has a solution iff $U=c_1\,(q+c_2)^{-2}+c_0$ for some constants $c_2,c_1,c_0\in\R$.
\end{itemize}
\end{lemma}

\begin{proof}
\mbox{}\\
{\it Proof of (a.):}
Recall that the complexifier is the only constant of motion with respect to the flow generated by $C$ itself so that the homogeneous solutions to the first order PDE $\{C,G\}=F(C)$ are functions of $C$. Therefore, the most general solution $G$ with $\deg_{\,p} \,G= 2m+1$ is of the form 
\bq
G=F(C)\, g(p,q)+g_h(C)~.
\eq
Here, $g_h$ is any polynomial of $C$ whose degree is not exceeding $\deg_{\,C}\, F$ and $g=g(p,q)$ is a linear function in $p$ for which
\bq
1\overset{!}{=}\{C,g\}=p\,\frac{\partial\,g}{\partial q}-\frac{\partial\,g}{\partial p}\, U'~.
\eq
Since $U$ and $\frac{\partial g}{\partial p}$ are independent of $p$ this can be only solved if $g$ and $U'$ are constant in $q$.
\\
{\it Proof of (b.):}
Without loos of generality one can assume that $F_0=0$ because the Poisson bracket does not change when $C$ is shifted by a constant. Suppose $g=p\, h_1(q)+h_0(q)$ then
\bq
G=\frac{F(C)}{C} g(p,q)+g_h(C)
\eq
is a generic solution that has the right degree in $p$ iff
\bq
C\overset{!}{=} \{C,g\}=p(p\,h_1'+h_0')-h_1 U'~.
\eq
This implies that $h_0$ is constant, $h_1=\frac{1}{2}( x+ c_2)$ and 
$2\,U= - (x+c_2) U'$ which is solved by $U=c_1\,(q+c_2)^{-2}+c_0$.
\end{proof}
Remember that lemma \ref{lem:vanishing} guarantees the existence of a non-zero integer $m_0<N$, for which $\deg_{\,C}\, \{C,H\}_{(2N)}=m_0$ and $\deg_{\,p}\, \{C,H\}_{(2N-1)}=2m_0\pm1$, and that the potential $U$ must be either linear in $q$ (for $2m_0+1$) or proportional to $q^{-2}$ (for $2m_0-1$)  according to lemma \ref{lem:U}. 
The degree in $q$ of the functions $f^{\mu}_{\nu}$ can be directly determined from equation \eqref{eqn:f}, that is,
\bq
\deg_{\,q}\, f^{n}_{n+2}=\deg_{\,q}\, \alpha - n
\eq
and 
\bq
\deg_{\,q}\, f^{n}_{n+2-2m}=\max\left[ \deg_{\,q}\, \alpha-n+m+ m \,\deg_{\,q}\, U' ,\,\deg_{\,q}\, V-n+m-1+ (m-1) \,\deg_{\,q}\, U' \right]~.
\eq
Suppose the coefficients $f^{\ast}_{2m_0+1}$ are not zero then they have to be constant as $f^{\ast}_{2m_0+2}=0$. In this case, lemma \ref{lem:U} states that also $ U'$ has to be constant which is why
 \bq
 0\overset{!}{=} \deg_{\,q}\,f^{2m_0-1}_{2 m_0+1}= \deg_{\,q}\,\alpha-(2m_0-1)
 \eq
and 
\bq
0\overset{!}{=} \deg_{\,q}\,f^{2m_0+1}_{2 m_0+1}
 =\max\left[ \deg_{\,q}\,\alpha-(2m_0+1)+1 ,\,\deg_{\,q}\, V-(2m_0+1)+1-1\right]~.
 \eq
This forces $\deg_{\,q}\,\alpha=2m_0-1$ and $\deg_{\,q}\, V=2m_0+1$.  Yet,
\bq
0\overset{!}{=} \deg_{\,q}\,f^{2m_0+3}_{2 m_0+1}
 =\max\left[2m_0-1-(2m_0+3)+2 ,2m_0 -(2m_0+3)+1\right]<0
 \eq
 leads to an inconsistency and thus $f^{\ast}_{2m_0+1}$ has to vanish which means that $f^{\ast}_{2m_0}$ must be constant. By using lemma \ref{lem:U} (b.) one can deduce that the potential $U$ is proportional to $q^{-2}$ and consequently $\deg_{\,q}\,\alpha$ must be equal to $2m_0-2$ and $\deg_{\,q}\,V$ must equal $2m_0$. Again
\bq
0\overset{!}{=} \deg_{\,q}\,f^{2m_0+2}_{2 m_0}<0
 \eq
leads to a contradiction so that all $f^{\ast}_{2 m_0}$ have to vanish. By repeating this argument one can finally show that $f^{\ast}_m$ must vanish for all $m>2$. Hence, $f^{\ast}_2$ is constant and $f^{n-1}_1=-f^n_2 x+c^{n-1}_1$ for some $c^{n-1}_1\in\R$ which can be only achieved if $\alpha$ is constant as well and $U$ is proportional to $x^{-2}$. Obviously, the brackets $\{H,C\}_{(2n)}$ and $\{H,C\}_{(2m)}$, $n,m\neq 0$, can only differ by an over-all constant so that already $\{H,C\}_{(3)}=0$. This proves the following theorem.
\begin{theorem}
Suppose the Hamiltonian and the complexifier are quadratic in the momentum $p$ then the only system that solves $\{C,H\}_{(2N+1)}=0$ for some $N>0$ and \eqref{eq:stab1} is the radial oscillator \eqref{radial_oscillator} and canonical conjugates thereof. Furthermore, $N$ equals 1.   
\end{theorem}
Note, all the equations used to prove this theorem are directly related to Poisson brackets and for this reason only hold up to canonical transformations. This point will be exploited heavily in the next section to derive a more general construction principle for stable coherent state system.    

\section{Adapted complexifiers for integrable systems}
\label{Adapted_Complexifier}

The preceding investigations have revealed that it is in general very hard to construct a complexifier adapted to the dynamics of a given model. Nevertheless, the derived criteria are form invariant under canonical transformation which opens the possibility to excess a wider class of models than those examined above. A common feature of the harmonic and the radial oscillator is their periodic motion which is why integrable systems that show a quasi-periodic motion seem to be especially promising candidates. This idea is to be elucidated in more detail in the subsequent section. In the first part a generalized construction principle for adapted complexifiers will be worked out using so-called action-angle coordinates and the Hamilton-Jacobi approach. This will then be tested on several examples.  

\subsection{Generalized construction principle}
\label{sec:Generalized_construction}

Throughout this section it will be assumed that the mechanical models in question are not explicitly time-depended and posses only a finite number $\mathfrak{f}$ of degrees of freedom. 

A \emph{first integral of motion}\footnote{The term `first integral' often refers to a global property while `constant of motion' is used in a more local context. Yet, the nomenclature is far from being unique; here the conventions of {\arnold} will be used.} is a $C^1$ function $f$ on the phase space $\phasespace$ that Poisson commutes with the Hamiltonian $H$, i.e $\{H,f\}=0$ on the entire phase space. Two functions $g,f\in C^1(\phasespace)$ are said to be in \emph{involution} if $\{g,f\}=0$ and (functionally) independent on a subset $U\subset\phasespace$ if the one-forms $\dif f$ and $\dif g$ are linearly independent on $U$. If $f$ and $g$ are independent then in particular $\dif f \wedge \dif g $ is non-zero on $U$.  
\begin{definition}[Integrable System]
A system with $\mathfrak{f}$ degrees of freedom is integrable if there exist $\mathfrak{f}$ first integrals $H_j$, $j=1,\cdots,\mathfrak{f}$, in involution that are independent on a dense subset of $\phasespace$.
\end{definition}
The name is motivated by the fact that such systems are integrable by quadratures, that is, they are solvable by a finite number of algebraic operations. This insight goes back to Liouville and was later enlarged by Arnold by the following (see e.g. \arnold):
If the system is integrable then the level sets 
\beq
\label{levelset}
M_{\mathbf{h}}=\{\mathfrak{m}\in\phasespace| H_j(\mathfrak{m})=h_j\}
\eeq
are smooth submanifold that are invariant under the phase flow generated by the Hamiltonian $H$. If $M_{\mathbf{h}}$ is compact and connected then it is diffeomorphic to the torus $T^{\mathfrak{f}}$ and the motion is conditionally periodic. This means that $M_{\mathbf{h}}$ has coordinates $\Theta_j$ which  parametrize the circles $S^1$ in $T^{\mathfrak{f}}$ and which evolve as  
\beq
\frac{\dif \mathbf{\Theta}}{\dif t}=\boldsymbol{\omega}(\mathbf{h})~.
\eeq
For this statement to hold it actually suffices that the first integrals are independent on $M_{\mathbf{h}}$. Furthermore, this angle coordinates can be used to parametrize the phase space (in the neighborhood of the invariant torus). Its conjugated momenta $I_j$ can be found by a canonical transformation that, by definition, leaves the symplectic structure invariant. So
\bq
\sum_{j=1}^{\mathfrak{f}} \dif p_j\wedge \dif q_j=\sum_{j=1}^{\mathfrak{f}} \dif I_j\wedge \dif \Theta_j~.
\eq
The parameters $(\mathbf{I},\mathbf{\Theta})$ are called action-angle coordinates in the literature and are widely used in classical perturbation theory (see e.g. \cite{arnold,Goldstein} for more details). Here,  their simple time-dependence is of interest which is given by 
\beq
\frac{\dif \mathbf{I}}{\dif t}=0 \quad\text{and}\quad
\frac{\dif \mathbf{\Theta}}{\dif t}=\boldsymbol{\omega}(\mathbf{I})~.
\eeq
This immediately shows that the complex parametrization,
\beq
\label{Action_coordinates}
\w_j:=\sqrt{I_j} \ee{i\, \theta_j}\quad\text{and}\quad\overline{\w}_j:=\sqrt{I_j} \ee{-i\, \theta_j}
\eeq
is `stable' in the sense that it obeys
\beq
\frac{\dif \w_j}{\dif t}=\{H, \w_j\}=i\,\omega_j(\mathbf{I})\, \w_j~.
\eeq
In addition, the pair $(\overline{\boldsymbol{\w}},\boldsymbol{\w})$ is canonical conjugated as
\bq
\{\overline{\w}_j, \w_k\}=i\,\delta_{jk}
\quad\text{and}\quad
\{\w_j, \w_k\}=\{\overline{\w}_j, \overline{\w}_k\}=0
\eq
and the polar decomposition of $\w_j$ reads 
\bq
\w_j=\frac{1}{\sqrt{2}}\left(Q_j-i\,P_j\right)
\eq
where 
\bq
Q_j=\sqrt{2\,I_j}\, \cos\theta_j
\quad\text{and}\quad 
P_j=-\sqrt{2\, I_j}\, \sin\theta_j~.
\eq
Since also $(\mathbf{P},\mathbf{Q})$ are canonical conjugated $\w_j$ can be written as a complexifier coordinate,
 \bq
 \w_j=\frac{1}{\sqrt{2}} \ee{-i\Lie_{\hamVF{C}}}\,Q_j~,
 \eq
with  $C=\frac{1}{2} \mathbf{P}\cdot\mathbf{P}$. But note, the parametrization given by \eqref{Action_coordinates} does not satisfy the criteria of theorem \ref{th:stabCCS} since the frequencies\footnote{In general: $\frac{\partial\omega_j}{\partial I_k}=\frac{\partial^2 H}{\partial I_k\,\partial I_j}=\frac{\partial\omega_k}{\partial I_j}$.} $\omega_j$ still depend on $I_j=\w_j \overline{\w}_j$. If the system is non-degenerate, which is the case if $\det\frac{ \partial \omega_j}{\partial I_k}\neq 0$, then the invariant tori are uniquely defined and independent of the initial choice of the coordinates $(\mathbf{\Theta},\mathbf{I})$. This implies that, no matter which action-angle coordinates are selected, the frequencies will always depend on $\mathbf{I}$. In other words \eqref{Action_coordinates} can only give rise to stable states if the system is degenerate. More specifically, if $H=\sum_j \omega_j I_j$ for constant $\omega_j$ then the associated coherent states are stable because $\frac{\partial \omega_j}{\partial I_k}=0$ for all $j,k$. 

Before bothering about degeneracy one has to solve the more practical problem of how to determine action variables in the first place. A very useful tool for that is the Hamilton-Jacobi formalism: Given an Hamiltonian $H$ as a function of $\mathbf{p}$ and $\mathbf{q}$, the goal is to find a function $S(q_1,\cdots,q_{\mathfrak{f}}, h_1,\cdots, h_{\mathfrak{f}})$ with $\det\frac{\partial^2S}{\partial q_j\partial H_k}\neq0$ and $\frac{\partial S}{\partial q_j}= p_j $ such that 
\beq
\label{eqn:H_J}
H(q_1,\cdots,q_{\mathfrak{f}}, \frac{\partial S}{\partial q_1},\cdots, \frac{\partial S}{\partial q_{\mathfrak{f}}})
=K(h_1,\cdots, h_{\mathfrak{f}})~.
\eeq
Note, $S$ generates a canonical transformation since 
\begin{gather}
\nonumber
\dif S=\sum_j (p_j \,\dif q_j+\Theta_j \,\dif h_j)\\
\label{eqn:generator}
\Rightarrow\quad
\sum_j \dif p_j \wedge \dif q_j=\dif h_j\wedge\dif\Theta_j
\end{gather}
with $ \Theta_j:=\frac{\partial S}{\partial h_j}$.
Suppose one can find such a solution $S$ to the first order PDE \eqref{eqn:H_J} then, by \eqref{eqn:generator}, one has also found $\mathfrak{f}$ independent constants of motion\footnote{We here use the term `constants of motion' to indicate that they are in general not globally defined.} $h_j$ in involution. If $\det\frac{\partial^2S}{\partial q_j\partial H_k}\neq0$ as required then $h_j=h_j(\mathbf{p},\mathbf{q})$ can be extracted out of the equation $\frac{\partial S}{\partial q_j}= p_j $ due to the inverse function theorem. Apart form that also the dynamics of $\Theta_j$ simplifies, i.e. 
\bq
\frac{\dif \Theta_j}{\dif t}=\{H,\Theta_j\}=\frac{\partial K}{\partial h_j}~.
\eq
By setting $K(h_1,\cdots, h_{\mathfrak{f}})=E=h_1$ the dynamics becomes especially simple, instead of $\mathfrak{f}$ constants of motion one now has $2\mathfrak{f}-1$ constants. This seems rather odd, notably if it would hold globally, as it effectively reduces the system to a one-dimensional free-particle. Astonishingly, the Hamilton-Jacobi equation \eqref{eqn:H_J} with $K=E$ always has a \emph{local} solution\footnote{That means that a solution exist in a neighborhood of any generic point $\mathfrak{m}\in\phasespace$ on which $H$ is not extremal.} given by the classical action functional (see e.g. \cite{Goldstein, arnold, HilbertCourant}). This demonstrates that in the above construction the global properties are essential. On the other hand, one can also fix $K(h_1,\cdots, h_{\mathfrak{f}})=\sum_j h_j$ locally which proofs the lemma:
\begin{lemma}
\label{local_coordinates}
For each generic point $\mathfrak{m}\in\phasespace$ exist a neighborhood $U_{\mathfrak{m}}$ on which the functions 
\bq
\w_j:=\sqrt{h_j} \ee{i\omega_j\,\Theta_j} 
\quad\text{and}\quad
\overline{\w}_j:=\sqrt{h_j} \ee{-i\omega_j\,\Theta_j} 
\quad j=1,\cdots,\mathfrak{f}
\eq
are well-defined, functionally independent, canonical conjugated ($\{\overline{\w}_j,\w_k\}=i \omega_j \delta_{jk}$) and obey 
\bq
H=\sum_j \w_j\,\overline{\w}_j
\quad\text{and}\quad
 \{H,\w_j\}= i\omega_j\,\w_j~.
\eq
\vspace*{-15pt}
\mbox{}\\
Similarly one can introduce coordinates $Q_j:=\sqrt{\frac{h_j}{2\omega_j}} \cos \omega_j\Theta_j$ and $P_j:=-\sqrt{\frac{h_j}{2\omega_j}} \sin \omega_j\Theta_j$ which are well-defined, functionally independent and canonical conjugated on $U_{\mathfrak{m}}$. In terms of this coordinates $\w_j$ can be rephrased as 
\bq
\w_j=\frac{1}{\sqrt{2\omega_j}}\,\ee{-i\Lie_{\hamVF_C}} Q_j
\quad\text{with}\quad 
C=\frac{1}{2}\,\mathbf{P}^2~.
\eq
\end{lemma}
As already mentioned the $\w_j$'s will in general not define good coordinates on $\phasespace$ but can only be defined locally. Below it will be demonstrated on the simplest example of a free particle what happens if the global properties are ignored (see section \ref{sec:free_particle}). In contrast to that the complex parametrization defined through proper action-angle coordinates is at least well defined in the neighborhood of the whole torus.\footnote{In fact, this parametrization only breaks down at separatrices where $M_{\mathbf{h}}$ even ceases to be a manifold. Such separatrices divide the phase space into several regions on which the level sets may have different properties. An example for that is the mathematical pendulum where the phase space divides into an oscillatory and a rotationally branch (see e.g. \arnold).}

In order to understand better what are the necessary criteria for the existence of well-defined parametrizations $\w_j$ it is a good idea to investigate further the relation between the $\w$-variables and those obtained by the complexifier method. The question is whether $z_j=\ee{i\Lie_{\hamVF_C}} q_j $ can be obtained by lemma \ref{local_coordinates} given that $\mathbf{z}$ define good coordinates, i.e. they are everywhere defined and functionally independent, and given that 
\beq
\label{stable}
\{H,z_j\}=i f_j(\mathbf{z})~.
\eeq
In general, $z_j$ will \emph{not} be equal to $\w_j$ since  $\w_j$ is canonically conjugated to its complex conjugate $\overline{\w}_j$ but
\bq
\{\overline{z}_j,z_k\}=\{\ee{i\Lie_{\hamVF_C}} q_j,\ee{-i\Lie_{\hamVF_C}} q_k\}
=\ee{i\Lie_{\hamVF_C}} \{q_j,\ee{-2i\Lie_{\hamVF_C}} q_k\}
\eq
is generically not even constant (see section \ref{sec:first_solution} for an example). The momenta conjugated to $z_j$ and $\overline{z}_j$ are $\Pi_j=\ee{-i\Lie_{\hamVF_C}} (p_j + u_j(q_j))$ and $\Pi_j=\ee{i\Lie_{\hamVF_C}} (i\,p_j + u_j(q_j))$ where $u_j$ is any function that should only depend on $q_j$ to ensure that the momenta are Poisson commuting. Note, that the $u_j$ are not arbitrary but, due to the stability criterion \eqref{eqn:classco}, must be such that   
\beq
\label{reality1}
H=\sum_j \left(i\,f_j(\mathbf{z}) \,\Pi_j+ g_j(\mathbf{z})\right)=\sum_j \left(-i\,\overline{f}_j(\overline{\mathbf{z}}) \,\overline{\Pi}_j+ \overline{g}(\overline{\mathbf{z}})\right)
\eeq 
with $g_j(\mathbf{q})=\frac{1}{\mathfrak{f}}\, g(\mathbf{q})- u_j(q_j)$. In general the functions $f_j,\, u_j$ and $g_j$ can depend on complex parameters, which is why $\overline{f_j(\mathbf{z})}=\overline{f}_j(\overline{\mathbf{z}})$.

Even though the $z$ and $\w$ parametrization are generically different they are still closely related. To see this let us first derive a set of $\mathfrak{f}$ constants of motion of \eqref{reality1}. A  function $H_k$ is a constant of motion of \eqref{reality1} iff 
\bq
\{\ee{i\Lie_{\hamVF_C}} H, \ee{i\Lie_{\hamVF_C}} H_k\}=
\{\sum_j f_j(\mathbf{q})\, p^j+g(\mathbf{q}), F_k\}\overset{!}{=}0
\eq
where $F_k:=\ee{i\Lie_{\hamVF_C}} H_k$. Additionally, the maps $F_k$ must be in involution which suggests the ansatz $F_1:=\ee{i\Lie_{\hamVF_C}} H$ and $F_j=F_j(\mathbf{q})$ for $j=2,\cdots,\mathfrak{f}$. This yields a system of $\mathfrak{f}-1$ linear, first order partial differential equations (PDE)
\beq
0\overset{!}{=} \sum_j f_j(\mathbf{q})\,\frac{\partial F_k}{\partial q_j}
\eeq
that can be solved locally by the method of characteristics (see appendix \ref{app:B}). A local solution of such an equation depends on a set of $\mathfrak{f}$ arbitrary initial functions so that it should in general not be problematic to find $\mathfrak{f}-1$ functionally independent ones at least for a sufficiently small neighborhood. Yet in practice, it can be very hard to actually determine them and it might be easier to use the Hamilton-Jacobi formalism. In the following, let 
\beq
F(\mathbf{p},\mathbf{q}):=\ee{i\Lie_{\hamVF_C}} H= \sum_j (i\, f_j(\mathbf{q})\, p^j + g_j(\mathbf{q}))
\eeq
and
\beq
\overline{F}(\mathbf{p},\mathbf{q}):=\ee{-i\Lie_{\hamVF_C}} H= \sum_j(- i\, \overline{f}_j(\mathbf{q})\, p^j + \overline{g}_j(\mathbf{q}))~.
\eeq
As before, we want to find maps $S(q_1,\cdots,q_{\mathfrak{f}}, F_1,\cdots, F_{\mathfrak{f}})$ and $\tilde{S}(q_1,\cdots,q_{\mathfrak{f}}, \overline{F}_1,\cdots, \overline{F}_{\mathfrak{f}})$ that satisfy $\frac{\partial S}{\partial q_j}=\frac{\partial \tilde{S}}{\partial q_j}=p_j$,
\begin{align}
\label{eqn:H_J2}
F(q_1,\cdots,q_{\mathfrak{f}}, \frac{\partial S}{\partial q_1},\cdots, \frac{\partial S}{\partial q_{\mathfrak{f}}})
=\sum_j F_j
\quad\text{and}\quad
\overline{F}(q_1,\cdots,q_{\mathfrak{f}}, \frac{\partial \tilde{S}}{\partial q_1},\cdots, \frac{\partial \tilde{S}}{\partial q_{\mathfrak{f}}})
=\sum_j \overline{F}_j~. 
\end{align}
If $\det\frac{\partial^2 S}{\partial q_j\partial F_k}\neq 0$ and $\det\frac{\partial^2 \tilde{S}}{\partial q_j\partial \overline{F}_k}\neq 0$ then $S$ and $\tilde{S}$ are generators of canonical transformations that map $(\mathbf{p},\mathbf{q})$ on the conjugated pairs $(\mathbf{F},\mathbf{\Phi})$ and $(\mathbf{\overline{F}},\boldsymbol{\overline{\phi}})$ respectively where $\frac{\partial S}{\partial F_j}:=\Phi_j$ and $\frac{\partial \tilde{S}}{\partial \overline{F}_j}:=\overline{\phi}_j$. In these variables the reality condition \eqref{reality1} can be replaced by 
\beq
\label{reality2}
\ee{-i\Lie_{\hamVF_C}} F_j=i\,f_j(\mathbf{z}) \,\Pi_j+ g_j(\mathbf{z})=-i\,\overline{f}_j(\overline{\mathbf{z}}) \,\Pi_j+ \overline{g}_j(\overline{\mathbf{z}})=\ee{i\Lie_{\hamVF_C}} \overline{F}_j:=h_j~.
\eeq
Together with the angles 
\beq
\label{new_angle}
\Theta_j:=\frac{1}{2}\left(\ee{-i\Lie_{\hamVF_C}} \Phi_j +\ee{i\Lie_{\hamVF_C}} \overline{\phi}_j\right)
\eeq
the action variables $h_j$ build a new canonical pair $(\mathbf{h},\mathbf{\Theta})$ because
\begin{align*}
\{h_j,\Theta_k\}=\frac{1}{2}\left[ \ee{-i\Lie_{\hamVF_C}} \{F_j, \Phi_k\}
+ \ee{i\Lie_{\hamVF_C}} \{\overline{F}_j, \overline{\phi}_k\}\right]=\delta_{jk}~,
\end{align*}
which can be used to construct a $\w$-parametrization as in lemma \ref{local_coordinates}. Since 
 $\Phi_j$ and $\overline{\phi}_j$ only depend on $\mathbf{q}$ and $\mathbf{F}$ and $\overline{\mathbf{F}}$  the new angle coordinates are of the form
\beq
\label{theta}
\Theta_j=\frac{1}{2}\left[\Phi_j(\mathbf{z},\mathbf{h})+\overline{\phi}_j(\overline{\mathbf{z}},\mathbf{h})\right]~.
\eeq
It still remains to check whether $\w_j=\sqrt{h_j}\,\ee{i\omega_j\,\Theta_j}$ are good coordinates that are everywhere defined and independent. For this it is unavoidable to solve the Hamilton-Jacobi equations \eqref{eqn:H_J2}. This is particularly easy if the Hamilton-Jacobi equation is completely separable, that is,  if $F$ is such that $F=\sum_j F_j(q_j,\frac{\partial S_1}{\partial q_j})$. In this case, $S=\sum_j S_j(q_j,F_j)$ and $\tilde{S}=\sum_j \tilde{S}_j(q_j,\overline{F}_j)$ with 
\begin{align}
\label{ansatzHJ}
\begin{split}
S_j(q,F)&=c_j(F)+\int_{q_0}^q\! \dif x\;\frac{F_j-g_j(x)}{i f_j(x)}\\ 
\tilde{S}_j(q,\overline{F})&=\tilde{c}_j(\overline{F_j})-\int_{q_0}^q\! \dif x\;\frac{\overline{F_j}-\overline{g_j}(x)}{i \overline{f_j}(x)}
\end{split}
\end{align}
are solutions of \eqref{eqn:H_J2}. 
Here, $c_j$ and $\tilde{c}_j$ are arbitrary functions which will be set to zero in the following. To simplify the notation also the label $j$ will be left away. From ansatz \eqref{ansatzHJ} one can deduce that 
\beq
\Phi(q)=-i\int_{q_0}^q\dif x\,\frac{1}{f(x)} 
\quad\text{and}\quad
\overline{\phi}(q)=\overline{\Phi}(q)=i\int_{q_0}^q\dif x\,\frac{1}{\overline{f}(x)}
\eeq
which yields 
\bq
\Theta= -\frac{i}{2}\left[\int_{z_0}^z\dif x\,\frac{1}{f(x)} -\int_{\overline{z}_0}^{\overline{z}}\dif x\,\frac{1}{f(x)} \right]~.
\eq
Thus, 
\beq
\label{difTheta}
\dif \Theta =\frac{-i}{2}\left( f(z)^{-1}\dif z- \overline{f(z)}^{\,-1}\dif \overline{z}\right)
\eeq
is well-defined for all points $\mathfrak{m}\in\phasespace$ on which $f(\mathfrak{m})=f(q)\neq0$. Moreover, since the Hamiltonian is real it will in most cases be a function of $z\overline{z}$ so that 
\beq
\label{wedgew}
\dif \w\wedge\dif \overline{\w}=\frac{i\omega}{2}\dif\Theta\wedge\dif h
=\frac{\omega}{4}\left(f(z)^{-1}\,\frac{\partial h}{\partial \overline{z}}+\overline{f(z)}^{\,-1}\,\frac{\partial h}{\partial z}\right)\dif z\wedge\dif \overline{z}
\eeq
is most likely non-vanishing and $\w=\sqrt{h}\ee{i\omega \Theta}$ will provide a good parametrization if $z$ does. 

The last thing which remains to check is whether the level sets \eqref{levelset} are compact, that is, whether $(\mathbf{\Theta},\mathbf{h})$ are proper action variables. Note, that if $\Theta$ is not a proper angle then $\w$ is multivalued and therefore the parametrization will break down on the points where $\Theta(\mathfrak{m}')=\Theta(\mathfrak{m})+n\pi$. Concluding, if \eqref{difTheta} and \eqref{wedgew} are well-defined nowhere vanishing differential forms then $\mathbf{\Theta}$ is most likely a parametrization of an invariant torus (for a counter example see section \ref{sec:free_particle}).

Even though the above considerations do not prove that integrability and compactness of the level sets is a necessary condition for stable coherent states it nevertheless uncovers a close relation between existence of (global) action-angle coordinates and those states. In particular, one can only hope to find a good global $z$-parametrization if the system is integrable, has compact level sets and a degenerate dynamics ($\det\frac{ \partial^2 H}{\partial I_j \partial I_k}=0$).

\subsection{Examples}
\label{sec:examples}

To illustrate the construction principle described in section \ref{sec:Generalized_construction} and to emphasize its advantages and drawbacks, some examples will be discussed. 

\subsubsection{Radial oscillator}

Apart from the harmonic oscillator the radial oscillator with Hamiltonian 
\bq
H=\frac{p^2}{2}+ \frac{\omega^2}{2} q^2+\lambda q^{-2}
\eq
and phase space $\phasespace=\{(p,q)\in\R^2|q\neq 0\}$ was the only one-dimensional model with stable coherent states found in section \ref{sec:first_solution}. To obtain stable states one has to consider the complexifier
\bq
C=\frac{1}{\omega}\left(\frac{p^2}{2}+\lambda q^{-2}\right)
\eq
for which 
\bq
\ee{i\Lie_{\hamVF_C}} H= i\,\omega \,p\,q+ g(q):=F
\eq
and 
\bq
z:=\ee{-i\Lie_{\hamVF_C}} q=\left[\left(q-\frac{i}{\omega}\,p\right)^2-2\frac{\lambda}{\omega^2} q^{-2}\right]^{\frac{1}{2}}~.
\eq
Starting with that, it is possible to construct a $\w$-parametrization displaying all the properties mentioned in lemma \ref{local_coordinates} by applying the Hamilton-Jacobi method with generating function  
\bq
S=-i\int_{q_0}^q \dif x\, \frac{1}{\omega x}\left(F-g(x)\right)~.
\eq
One easily verifies $\frac{\partial^2 S}{\partial q\,\partial F}=-\frac{i}{\omega q}\neq0 $ for all $\mathfrak{m}\in\phasespace$, $\frac{\partial S}{\partial q}= p$ and $F(q,\frac{\partial S}{\partial q})=F$. The angles, which are canonically conjugated to $F$ and $\overline{F}$, are 
\beq
\label{eqn:Int1}
\Phi=\frac{\partial S}{\partial F} = -\frac{i}{\omega} \ln x
\quad
\text{and}
\quad
\overline{\Phi}=\frac{\partial \overline{S}}{\partial\overline{F}} = \frac{i}{\omega} \ln x 
\eeq
and the angle conjugated to $h$ is 
\bq
\Theta:= - \frac{i}{2\omega}\left(\ln z-\ln \overline{z}\right)
\eq
This shows that the $\w$-parametrization is given by
\bq
\w:=\sqrt{h} \ee{i\omega \phi} 
=\sqrt{\frac{h\,z}{\overline{z}}}
= \sqrt{\frac{h}{\overline{z}\,z} }\,z=\sqrt{\frac{\omega\, h/2}{\sqrt{h^2/\omega^2-2\lambda}}}\; z
\eq
where the last equality follows from  
\begin{align*}
z\,\overline{z}=\left[\left(q^2-\frac{p^2}{\omega^2}\right)^2+4\frac{\lambda}{\omega^2}
\left(\frac{\lambda}{\omega^2}\, q^{-4} +\frac{p^2 q^2}{\omega^2}-1\right)\right]^{\frac{1}{2}}
=\frac{2}{\omega} \left[\frac{h^2}{\omega^2}-2\lambda\right]^{\frac{1}{2}}~.
\end{align*}
Of course, the functional dependence of $\w$ and $\overline{\w}$ on the original parametrization $(p,q)$ is much more complicated than for $z$ and $\overline{z}$. However, their algebraic properties are nicer. By construction $\{\overline{\w},\w\}=i\omega$, while for $z$ and $\overline{z}$ holds
\bq
\{\overline{z},z\}=\{\sqrt{\frac{|z|^2}{h}}\,\overline{\w},\sqrt{\frac{|z|^2}{h}}\,\w\} =\frac{4\,i}{\omega^3 |z|^2} H~.
\eq
Since the Hamiltonian is strictly positive on the phase space also $\frac{h^2}{\omega^2}-2\lambda$ must be grater than zero. Because of that and since $z$ and $\overline{z}$ are well-defined and independent, 
$\w$ and $\overline{\w}$ are everywhere defined and independent. Besides that, the level sets are compact.

The complex parametrization $\w$ can be as well brought into the complexifier form defining
\bq
Q:=\sqrt{\frac{ h/4}{\sqrt{h^2/\omega^2-2\lambda}}}\;\Re(z)
\quad\text{and}\quad
P:=-\sqrt{\frac{ h/4}{\sqrt{h^2/\omega^2-2\lambda}}}\;\Im(z)
\eq
and $C=P^2/2$. Then $w= \frac{1}{\sqrt{2 \omega}}\ee{-i\Lie_{\hamVF_C}} Q=  \frac{1}{\sqrt{2 \omega}} (Q-iP)$. 
\\[10pt]
Instead of starting with the $z$-parametrization it is equally well allowed to construct directly the action-angle coordinates of the model. With the generating function 
\bq
\tilde{S}(q,h)=\int_{q_0}^q\! \dif x\; \sqrt{2(h-V(x))}
\eq
the Hamilton-Jacobi approach leads to $H(q,\frac{\partial S}{\partial q})=h $, $\frac{\partial S}{\partial q}=p$ and 
\begin{align}
\label{eqn:Int2}
\tilde{\Phi}:=&\frac{\partial\tilde{S}}{\partial h}= \int_{q_0}^q \frac{\dif x} {\sqrt{2(h-V(x))}}
= \frac{1}{2\omega } \arcsin\frac{y}{\Omega}
\end{align}
where $\Omega^2=h^2/\omega^2- 2\lambda$, $y=\omega q^2-h/\omega$ and $q_0$ is fixed such that $y_0=1$ for simplicity. Since $\arcsin x=-i\ln(ix+\sqrt{1-x^2})$, the $\w$-parametrization, which one obtains from this ansatz, is equal to the above up to a complex phase depending on the value of $q_0$. Here, $\tilde{\w}=\ee{i\pi/4} \w$. Yet, the integration of \eqref{eqn:Int2} is more involved than that of \eqref{eqn:Int1}

\subsubsection{Free particle}
\label{sec:free_particle} 

The free particle with Hamiltonian $H=\frac{p^2}{2}$ is a good example for what can go wrong if the global properties are ignored. Locally the Hamilton-Jacobi equation, $H(q,\frac{\partial S}{\partial q})=h$, is solved by the generating function  
\bq
S(q,h)=\sqrt{2\,h} \;q 
\eq
with conjugated `angle' 
\bq
\Theta:=\frac{\partial S}{\partial h}= \frac{q}{\sqrt{2 h}}=\frac{q}{p}
\eq
which is only well-defined for $p\neq 0$. This yields
\bq
\w:=\sqrt{h} \ee{i\omega \Theta}= p(\cos q/p+i\sin \omega p/q)=\frac{1}{\sqrt{2\omega}}\ee{-i\Lie_{\hamVF_C}} Q
\eq
where $Q=\frac{1}{\sqrt{2\,\omega}}\,p\cos \frac{q}{p}$, $P= -\frac{1}{\sqrt{2\,\omega}}\,p\sin\frac{q}{p}$, and $C=P^2/2$. Locally, the differentials $\dif \w$ and $\dif\overline{\w}$ are well-defined and 
\bq
\dif p\wedge\dif q=\dif h\wedge\dif \Theta=\dif P\wedge\dif Q=\frac{1}{i \omega} \dif\w\wedge\dif\overline{\w}
\eq
is non-degenerate. But it is not possible to extend that to all of $\mathcal{M}$ since, firstly, the whole parametrization is ill-defined for $p=0$ and, secondly, the coordinates $(\w,\overline{\w}, P,Q)$ are multivalued, that is, $\w(q,p)=\w(q',p)$ for $\omega\frac{q}{p}= 2\pi n + \omega \frac{q'}{p}$ and $n\in\N$. Even more severe, the motion of $\w$ seems to be periodic:
 \bq
 \w(t)=\w(p_0,q(t))=p_0[\cos(\omega t+q_0/p_0)+i\sin(\omega t+q_0/p_0)]~.
 \eq
 All these problems arise from $\Theta$ not being a proper angle since the level sets $h=\mathrm{const.}$ are not compact. On the other hand, for $\omega<<q/p$ the parametrization is not to bad since up to second order
 \bq
 P\approx -p+{\cal O}(\omega^2)
 \quad\text{and}\quad
 Q=-\omega q+{\cal O}(\omega^3)~.
 \eq
Another interesting property of $\w$ is that on the level-sets, i.e. for constant momenta, $\w$ are the generators of the $\ast$-algebra of quasi-periodic functions\footnote{This are functions $f$ of the form $f(x)=\sum_{j\in{\cal I}} f_j \ee{i k_j x}$ for some finite label set $\mathcal{I}$ and $k_j\in\R$. Here, $k_j=\omega_j/p_0$.}. The spectrum of this algebra constitutes the so-called Bohr compactification that lies at the heart of Loop Quantum Cosmology (see e.g. \LQCreview).

\subsubsection{Anharmonic oscillators}

For a generic one-dimensional Hamiltonian of the form $H:=\frac{p^2}{2}+V(q)$ one can always solve the Hamilton-Jacobi equation locally by choosing
\beq
\label{Allgemeiner_Anstaz}
S=\int^q_{q_0}\dif x \sqrt{2(h-V(x))}~.
\eeq
To conclude, we will discuss two examples of an anharmonic oscillator that are widely studied in the literature: the P\"oschl-Teller potential \cite{Poeschel}
\bq
V_{P.T.}(q)=-\frac{\lambda(\lambda+1)}{2}\, (\cosh(q))^{-2}~,
\eq
which is an effective potential to describe vibrations in diatomic molecules and for which the Schr\"odinger equation is solvable, and the quartic potential 
\bq
V_4(q)=\frac{\omega^2}{2} q^2+\frac{\lambda^2}{2} q^4~,
\eq
which is a standard example in perturbation theory. 

For large $q$ the P\"oschel-Teller potential is exponentially suppressed since for large q  $\cosh(q)\approx e^{|q|}/2$ and hence the level sets are compact. Moreover, a local $\w$-parametrization is conceivable. Ansatz \eqref{Allgemeiner_Anstaz} yields
\begin{align*}
\Theta:=\frac{\partial S}{\partial h}&=\int_{0}^q \frac{\dif x\; \cosh x}{\sqrt{2h \cosh^2 x+\lambda(\lambda+1)}}\\
=&(2h )^{-1/2}\int_{0}^q \frac{\dif (\sinh x)}{\sqrt{\sinh^2 x+\Omega^2}}\\
=& (2h )^{-1/2} \;\sinh^{-1}\left(\frac{\sinh q}{\Omega}\right)~,
\end{align*}
where $\Omega^2=\frac{\lambda(\lambda+1)}{2h}+1$, and 
\bq
\w=\sqrt{h} \ee{i\omega \Theta}=\frac{\sqrt{h}}{\Omega}\left[\sinh q+\sqrt{\Omega^2+(\sinh q)^2}\right]^{\frac{i}{\sqrt{2h}}}~.
\eq
In the last equation the identity $\sinh^{-1} q=\ln[q+\sqrt{1+q^2}]$ was used. The corresponding $P,Q$-variables are even more complicated so that the quantization via the complexifier method can cause severe problems.

While the motion for the quartic potential is also bounded, the angle $\Theta$ is even harder to determine. In fact,
\bq
\Theta:=\int_{q_0}^q \frac{\dif x}{\sqrt{2h-\omega^2 x^2-\lambda^2 x^4}}
\eq
is an incomplete elliptic integral of first kind\footnote{For more information see e.g. \cite{Abramowitz:mf1970, groebener}} that \emph{is not} solvable in terms of elementary functions. By replacing $x=-b\cos\phi$ and
 \bq
 2h-\omega^2 x^2-\lambda^2 x^4=\lambda^2(x^2+a^2)(b^2-x^2),
 \eq
where $a^2-b^2=(\omega/\lambda)^2$ and $b^2 \omega^2+b^4\lambda^2=2h$, this integral can be brought into the so-called Legendre canonical form 
 \beq
 \Theta=\frac{1}{\lambda}\int_{q_0}^q \frac{\dif x}{\sqrt{(x^2+a^2)(b^2-x^2)}}
 =\frac{k}{\lambda b} F(\phi,k) +c~.
 \eeq
Here, $k^2=\frac{b^2}{b^2+a^2}$, $c$ is a constant and 
 \bq 
 F(\phi,k)=\int_0^{\sin \phi}\frac{\dif \xi}{\sqrt{(1-\xi^2)(1-k^2\xi^2)}}
 =\int_0^{ \phi}\frac{\dif \psi}{\sqrt{1-k^2 \sin^2\psi}}~. 
 \eq
A literal quantization of the corresponding complexifier seems hopeless. On the other hand the quartic anharmonic oscillator can be treated, classically as well as quantum mechanically, by a perturbation theory. Especially the classical perturbation theory makes heavy use of action-angle coordinates of the unperturbed harmonic oscillator. It therefore seems much more promising to study the $\w$-parametrizations in this context, however, this would go beyond the scope of this paper.  
 
\section{Discussion}

It was found that a necessary criterion for the existence of stable coherent states is that the classical evolution of the variable $\mathbf{z}=\ee{-i\Lie_{\hamVF_C}} \mathbf{q}$ depends only on $\mathbf{z}$ itself and not on its complex conjugate. However, it is not possible to determine other models than the harmonic and radial oscillators from ansatz \eqref{ansatz2a} that allow for such parametrizations. This issue was circumvented by invoking the Hamilton-Jacobi formalism. In doing so, one can, indeed,  construct local complex coordinates for almost any system that display the desired properties and are related to the complexifier formalism. In general, these can, however, not be extended globally what was explained exemplarily for the free particle. Another issue arises from the fact that the resulting complexifiers and parametrizations are in general not analytic functions  of the `old' variables $(p,q)$ (see section \ref{sec:examples}) that potentially causes severe problems when quantizing. 
\\

Let us speculate a bit more about the last point: One essential ingredient for quantization is the choice of a polarization. Very roughly, the polarization says how to split the classical phase space $\phasespace$ of dimension $2\,\mathfrak{f}$ into an $\mathfrak{f}$-dimensional submanifold ${\cal P}$ whose elements are then represented as multiplication operators on an appropriate Hilbert space. Of course, for $\phasespace=\cotb{\config}$ the natural choice of ${\cal P}$ is the configuration space $\config$, which corresponds  to the usual Schr\"odinger quantization. In the case of the harmonic oscillator we already mentioned implicitly another possible polarization, namely, that defined by the complex parameter $a=q-i\,p$. This leads to the Hilbert space $\mathfrak{H}^2(\C,\dif\nu)$ of holomorphic square integrable function on which the annihilator $\hat{a}$ acts by multiplication. The resulting quantum theory is unitary equivalent to the Schr\"odinger quantization and the transformation between the two representations is given by the Segal-Bargann transform \eqref{eqn:S-B_transform}. This close relation between coherent states and complex polarizations is much more generic and often used in geometric quantization.   

As was mentioned in section \ref{sec:CCS}, the complexifier approach follows exactly the above line of thoughts: It intends to provide a link between an arbitrarily chosen complex K\"ahler structure and a quantization on the corresponding polarizations. This relation was examined in great detail in \cite{HallKirwin1,Kirwin1, Kirwin2,HallKirwin2} focussing especially on unitary equivalence of the resulting quantum theories. An interesting result of \cite{Kirwin1} is that for compact Lie groups $G$ with phase space $\cotb{G}$ only complexifiers of the form $\hat{C}=\frac{\Delta}{2}$ give rise to a quantization that is unitary equivalent to the Schr\"odinger one with Hilbert space $L^2(G,\dif x)$. Here, $\Delta$ is the Laplacian of the group and $\dif x$ is the Haar measure. But formally $\Delta$ can be identified with the square $\hat{p}^2$ of the original momentum operators while the complexifiers obtained for the various examples in section \ref{sec:examples} show a much more complicated functional dependence on $p$. Thus, one should not expect that a quantum model obtained by directly quantizing the systems in the new $\w$-parametrization is equivalent to the Schr\"odinger representation. Even though it is, of course, tempting to try the first due to the non-analyticity of $C$, which one often encounters. 

Apart from the complicated structure of the complexifier it is also questionable whether the above strategy can be applied if the $\w$-parametrization is ill-defined. All of these are interesting questions which deserve a further investigation but go beyond the scope of this paper. 
\\

Remember, the original motivation for investigating the stability of coherent states originated from the heuristic implementation of the Hamiltonian constraint in LQG by the formal expression \eqref{eqn:stab1}. Yet, the formalism developed in section \ref{Adapted_Complexifier} can blow up already for `simple' one-dimensional models. The technical obstacles encountered in the finite dimensional examples of this paper are of course even more challenging in QFT ?s such as LQG. On the other hand, the formalism 
developed in this paper in principle directly applies to symmetry reduced models of LQG such as LQC. In fact there are two observations which support the assumption that the above formalism is conceivable in LQC: The first hint is the appearance of the Bohr compactification, which plays a crucial role in the quantization of LQC, in the stable state system of the free particle. The second hint is more vague, namely, the stable parametrizations $z,\overline{z}$ of the radial oscillator form an $\mathfrak{sl}(2,\R)$ algebra together with the Hamiltonian $H$. But this algebra also shows up in the modified dynamics of deparametrized, homogeneous and isotropic models.

{\acknowledgments
 T.T. thanks Jos\'e Mour\~ao, Joao Nunes and William Kirwin for several fruitful discussions.
This project was supported in part by funds from the Emerging Field Initiative of the Friedrich-Alexander-University Erlangen N\"urnberg to the Emerging Field Project ``Quantum Geometry. A.Z. acknowledges financial support by the grant of Polish Narodowe Centrum Nauki 
 nr 2012/05/E/ST2/03308.
}

\newpage
\appendix
\section{Stable complexifier states for the radial oscillator}
\label{app:A}
\label{sec:radOsc}

\subsection{Classical and Quantum Properties of the Radial Oscillator}
\label{ssec:Class/QuantProp}

Note that the potential of the radial oscillator,
\begin{align*}
V(q)=\frac{m\omega^2}{2}q^2+\lambda q^{-2}~,
\end{align*}
with mass $m$, frequency $\omega$ and coupling constant $\lambda$ diverges in the limit $q\to 0$. Therefore, the classical motion is confined to the positive (or negative) axes. The equation of motion can be solved, for example, by integrating the law of conversation of energy via a separation of variables. This yields
\begin{align*}
q(t)=\pm\left[\frac{E}{m\omega^2}+\frac{\sqrt{\gamma}}{\sqrt{m}\omega} \sin{(2\sqrt{m}\omega t+\phi)}\right]^{\frac{1}{2}}
\end{align*} 
where $\gamma=\frac{E^2}{\sqrt{m}\omega}^2-2\lambda$ and the phase $\phi$ is determined by the initial condition. 

From the form of the potential, one expects the quantized system to be discrete. In the limit $\lambda\to0$, the solutions should approach the solutions of the harmonic oscillator. Due to the barrier at $q=0$, only the odd solutions, i.e. those solutions which have a knot at zero, are allowed. For `big' $\lambda$ the solutions on the positive and negative axis should decouple. This can be checked explicitly by replacing the peak at $q=0$ through a box potential of width $d=2\sqrt{\frac{\lambda m}{\omega^2}}$ and height $V_0$. The transmission coefficient $T$ for such a system with energy $E$ is then given by 
\begin{align*}
T=\left(1+\frac{V_0^2\operatorname{sinh}^2(\sqrt{2(V_0-E)}d)}{4E(V_0-E)}\right)^{-1}~,
\end{align*}
which vanishes for $V_0\to\infty$. Therefore, also the quantum system can be restricted to the positive axis. The natural choice for a Hilbert space is the space $L^2(\mathbb{R}^+,\dif q)$ of square intgrable functions with respect to the Lebesgue measure $\dif{q}$ where $\hat{q}$ and $\hat{p}$ are represented as multiplication and derivative operator respectively. 

The corresponding eigenvalue equation,
\begin{gather*}
\hat{H}\,\psi(x)=E\,\psi(x)\nonumber\\
\implies \psi''(x)-\left(y^2+\frac{\alpha^2-1}{4}+2\epsilon\right)\psi(x)=0
\end{gather*}
with $x:=\sqrt{\frac{m\omega}{\hbar}}$, $\epsilon:=\frac{E}{\omega\hbar}$ and $\alpha^2:=8 m \lambda+1$ is a modified Whittaker equation (see \cite{Abramowitz:mf1970}, p. 505) whose solutions are given by a hypergeometric function of the first kind $M$, that is, 
\begin{align*}
\psi(x)= N\ee{-\frac{x^2}{2}}x^{\frac{1+\alpha}{2}}\,M\!\left[\frac{1}{4}(2+\alpha-2\epsilon),\frac{2+\alpha}{2},x^2\right]
\end{align*}
where $N$ is a normalization constant. To ensure that $\psi$ is an element of the Hilbert space, the following conditions must be satisfied:
\begin{itemize}
	\item $\lim_{x\to0}\psi(x)<\infty$ requires $\alpha\geq-1$. Otherwise $\psi$ can not be square-integrable.
	\item In order that the Hamiltonian is well-defined, i.e. $\|\hat{H}\psi\|<\infty$, $\psi$ needs to be at least once continuously differentiable for $x\to0$ (and twice everywhere else). This requires $\alpha^2\geq1\iff\lambda\geq0$.
	\item $\psi$ is square-integrable iff $\frac{1}{4}(2+\alpha-2\epsilon)=-n,\;n\in\mathbb{N}$. Thus the energy spectrum is discrete and equidistant
	\begin{align}
	\label{eqn:eigenvro}
	E_n=\omega\hbar(2n+1+\frac{\alpha}{2})
	\end{align}
In the limit $\lambda\to0$ the eigenvalues approach the odd energy levels of the harmonic oscillator as expected and $M$ reduces to an Hermite polynomial (\cite{Abramowitz:mf1970}, p. 505). 
\end{itemize}
For these specific values the hypergeometric function $M$ becomes proportional to a generalized Laguerre polynomial $\laguerre{n}{\frac{\alpha}{2}}{x}$, $n\in\N$, which build a complete orthogonal system in $L^2(\R^+,y^{\frac{\alpha}{2}}\ee{-y} \dif y)$. Thus, the normalized eigenstates are given by 
\begin{align}
\label{eqn:ROeigen}
\psi_n(x)=\sqrt{\frac{2 n!}{\Gamma\left(\frac{\alpha}{2}+1+n\right)}}\ee{-\frac{x^2}{2}}x^{\frac{\alpha+1}{2}}\laguerre{n}{\alpha/2}{x^2}~.
\end{align}

\noindent{\bf Remark:} Since the eigenvalue equation for the radial oscillator is a second order partial differential equation there exist of course two independent solutions. Here the second linear independent solution is
\begin{align}
\label{eqn:secdSolut}
\tilde{\psi}(x)= N\ee{-\frac{x^2}{2}}x^{\frac{1-\alpha}{2}}\,M\!\left[\frac{1}{4}(2-\alpha-2\epsilon),\frac{2-\alpha}{2},x^2\right]
\end{align}
These functions have to obey analogous normalization conditions: 
\begin{align*}
\alpha\leq1 \qquad 	E_n=\omega\hbar(2n+1-\frac{\alpha}{2})\;.
\end{align*}
In order that the Hamiltonian is well-defined on these states also requires $\alpha^2\geq1$. Therefore $\lambda$ has to be greater than zero. On the other hand, $\alpha$ is equal to $\pm\sqrt{1+8 m \lambda}$. Thus either $\alpha\geq1$ or else $\alpha\leq-1$. For $\alpha>0$ we get the solutions above, for $\alpha<0$ we obtain \eqref{eqn:secdSolut}.

\subsection{Coherent States}
\label{ssec:compl}

Recall, that the complex parametrization $z=\ee{-\imath\mathcal{L}_{\chi_C}}q$ associated to the complexifier 
\begin{align}
\label{eqn:Complro}
C=\frac{1}{m\omega}\left(\frac{1}{2}p^2+m \lambda q^{-2}\right)
\end{align}
is given by 
\begin{align}
\label{eqn:ccsan}
z=\sqrt{\frac{2}{m\omega}}\sqrt{a^2- m \lambda q^{-2}}
\end{align}
where $a=\frac{m\,\omega}{2}q-i\frac{p}{2}$. The Poisson relations for these variables are
\begin{gather*}
\{z,\bar{z}\}=-\imath\frac{4}{\omega^3 m^2}\frac{H}{z\bar{z}}~,\;\;
\{H,z\}=\imath\omega z\;\,\text{ and }\; \, \{H,\bar{z}\}=-\imath\omega \bar{z}~.
\end{gather*}
For the operators of the rescaled maps  
\begin{gather*}
L_3:=\frac{\omega}{2} H~, \:\,
L_-:=\frac{m\omega^2}{4\sqrt{2}}\, z^2
\;\,\text{ and }\;\, L_+:=\frac{m\omega^2}{4\sqrt{2}}\, \bar{z}^2
\end{gather*}
one recovers the commutation relations of $\su(1,1)$:
\begin{align}
\label{eqn:su11}
[\hat{L}_3,\hat{L}_{\pm}]=\pm \hat{L}_{\pm}\;\,\text{ and }\;\,  [\hat{L}_-,\hat{L^+}]=\hat{L}_3\;.
\end{align} 
Therefore, it is expected that the complexifier coherent states coincide with the coherent states defined by Barut and Giradello (see \cite{BarutGiradello1971}). 

The easiest way to find the associated coherent states 
\begin{align}
\label{eqn:CCS-Sy}
\psi_z(q)=\ee{-\frac{\hat{C}}{\hbar}}\delta_z(q)
\end{align}
is to express the convolution $\delta$ in terms of the eigenfunctions of $\hat{C}$ that are given in terms of  Bessel functions of the first kind $\bessel{\beta}{x}$ (see e.g. \cite{Abramowitz:mf1970}, pp. 358-364). That is, the eigenfunction to the eigenvalue $\hbar\omega c^2,\;c>0$ is 
\bq
\Phi_c(x)=\sqrt{c x}\,\bessel{\alpha/2}{cx}
\eq
where $x=\sqrt{\frac{m\omega}{\hbar}}$. These functions are not square-integrable on $\mathbb{R}^+$ as expected but define a convolution:
\begin{align}
\label{eqn:delta2}
\delta_y(x)=\int_0^{\infty}{\dif{c}\,\Phi_c(x)\Phi_c(y)}~.
\end{align}
Thus, one finds 
\begin{align}
\psi_z(x)&=\left[\int_0^{\infty}{\dif{c}\,c\ee{-\frac{c^2}{2\hbar}}\sqrt{x\,y}\,\bessel{\alpha/2}{cx}\bessel{\alpha/2}{cy}}\right]_{y\to z}\nonumber\\
\label{eqn:CCSro1}
	&=\exp{-\frac{x^2+z^2}{2}}\mbessel{\alpha/2}{xz}~.
\end{align}
Here, $\mbessel{\beta}{x}=\imath^{-\beta}\bessel{\beta}{\imath x}$ is the modified Bessel function of the first kind. These states can be also expressed in terms of the eigenfunctions of the Hamiltonian by utilizing the relation of Bessel functions and Laguerre polynomials (see \cite{Abramowitz:mf1970}, p.734). This yields:
\begin{align}
\label{eqn:CCSro2}
\psi_z(x)=\ee{-\frac{z^2}{4}}\sum_{n=0}^{\infty}(-1)^n\frac{\left(z/\sqrt{2}\right)^{2n+\frac{\alpha+1}{2}}}{\sqrt{n!\Gamma(n+\alpha/2+1)}}\Psi_n(x)~.
\end{align}
For $z\to z^2/\sqrt{2}$ and $x\mapsto x^2/2$ one recovers the coherent states defined by Barut and Giradello \cite{BarutGiradello1971} up to a constant factor as expected.  However, the states \eqref{eqn:CCSro2} are not  normalized with respect to the real measure $\dif{x}$. Instead, one finds 
\begin{align}
\|\psi_z\|_{L^2(\R^+)}=\frac{|z|}{2}\ee{-\frac{z^2+\bar{z}^2}{4}}\,\mbessel{\alpha/2}{\frac{|z|^2}{2}}~.
\end{align}

\noindent{\bf Properties:}
The states \eqref{eqn:CCSro2} are over-complete with respect to the measure 
\begin{align}
\label{eqn:CCSroru}
\dif{\mu(z)}=[\mbessel{0}{\frac{r^2}{2}}]^{-1}\frac{r^2}{4\pi}\mbesselt{\alpha/2}{\frac{r^2}{2}}\dif{r}\dif{\theta}
\end{align}
expressed in polar coordinates $z=r\ee{\imath\theta}\;r\geq 0,\phi\in[-\frac{\pi}{2},\frac{\pi}{2}]$. Here, $\mbesselt{\beta}{x}$ is a modified Bessel function of the second kind. The unusual integration domain of $\phi$ originates from the fact that the analytic continuation $q\to z$ for the radial oscillator is defined only on the upper complex plane. Moreover, the states are truly stable under the dynamics, i.e. , 
\begin{align*}
\ee{\imath\frac{\omega}{\hbar} t\hat{H}} \psi_z(x)
=\ee{\imath\omega/2}\psi_{z(t)}(x)~.
\end{align*}
Since \eqref{eqn:CCSro2} are the eigenvectors of the annihilation operator $\hat{z}$, they minimize the uncertainty of
\begin{align}
\hat{a}:=\frac{\hat{z}^2+(\hat{z}^2)^{\dagger}}{2}=-\frac{2}{\omega}\hat{C}+\hat{x}^2
\end{align}
and of
\begin{align}
\hat{b}:=\frac{\hat{z}^2-(\hat{z}^2)^{\dagger}}{2}=-\frac{1}{\omega}(\hat{x}\hat{p}+\hat{p}\hat{x})\;.
\end{align}


\section{Partial differential equations of first order }
\label{app:B}
 
Partial differential equations (PDE) emerge in all kinds of physical problems and are widely studied in mathematics. A detailed treatise would go beyond of this work. It is rather intended to give a rough overview of properties and solution techniques for quasi-linear, first order PDE mentioned in the main text. For the interested reader we recommend the very detailed book by Hilbert and Courant \cite{HilbertCourant} or the more elementary book by Cohen \cite{Cohen}.

A first order PDE is called quasi-linear if it is of the form 
\beq
\label{app:B1}
\sum_{j=1}^n a_j(\mathbf{x}, u) \frac{\partial u}{\partial x_j}=b(\mathbf{x},u)
\eeq 
where $a_j$ and $b$ are continuos differentiable functions of $\mathbf{x}:=(x_1,\cdots,x_n)$ and $u$. The equation is linear if $a_j$ and $b$ do not depend on $u$. A solution $u(\mathbf{x})$ to \eqref{app:B1} defines an $n$-dimensional surfaces, called \emph{integral surface}, whose tangential vectors $v_i=\frac{\partial u}{\partial x_i}$ satisfy \eqref{app:B1} at every point with coordinates $(\mathbf{x},u)$.
  
To solve \eqref{app:B1} it is sufficient to determine a family of $(n-1)$-parametric \emph{characteristic curves} $\mathbf{x}(s,t_1,\cdots t_{n-1}), u(s,t_1,\cdots t_{n-1})$ that obey
\bq
\frac{\dif x_j}{\dif s}=a_j(\mathbf{x},u) \quad\text{and}\quad \frac{\dif u}{\dif s}=b(\mathbf{x})
\eq
with initial values $\phi_j(t_1,\cdots t_{n-1})=x_j(0,t_1,\cdots t_{n-1})$ and $u(0,t_1,\cdots t_{n-1})=\chi_k(t_1,\cdots t_{n-1})$. For a given set of initial data on the $(n-1)$-dimensional initial manifold $C$ there exist a unique solution of \eqref{app:B1} iff the functional determinant
\bq
\Delta:=\det
\begin{vmatrix}
a_1 &\cdots& a_n\\
\frac{\partial x_1}{\partial t_1} &\cdots&\frac{\partial x_n}{\partial t_1}\\
\vdots & \ddots &\vdots\\
\frac{\partial x_1}{\partial t_{n-1}} &\cdots&\frac{\partial x_n}{\partial t_{n-1}}
\end{vmatrix}
\eq
does not vanish. If $\Delta$ is zero then solutions of the initial value problem can only exist if $C$ is a characteristic manifold, that means, that $C$ is generated by a family of $(n-2)$-parametric characteristic curves itself (see \cite{HilbertCourant} for details). In this case, there exist infinitely many solutions. 

A generic first order PDE can be always written in the form $H(\mathbf{x},\mathbf{p},u)=0$ where $p_j=\frac{\partial u}{\partial x_j}$. Suppose $H$ is twice continuously differentiable then solutions can be found by a  method similar to the above. Namely, one replaces the PDE by the system of ordinary differential equations (ODE)
\bq
\frac{\dif H}{\dif s}=\sum_{j=1}^n \frac{\partial H}{\partial x_j} \, \frac{\dif x_j}{\dif s}
+\sum_{j=1}^n \frac{\partial H}{\partial p_j} \, \frac{\dif p_j}{\dif s}
+\frac{\partial H}{\partial u}\,  \frac{\dif u}{\dif s}
\eq 
together with characteristic equations
\bq
 \frac{\dif x_j}{\dif s}=  \frac{\partial H}{\partial p_j} ~,
 \quad
 \frac{\dif u}{\dif s}=\sum_{j=1}^n \,p_j\,  \frac{\partial H}{\partial p_j} ~,
 \quad
  \frac{\dif p_j}{\dif s}=  -(\frac{\partial H}{\partial x_j}+\frac{\partial H}{\partial u}\, p_j) ~.
\eq
Note, that all these methods are local implying that existence of solutions only holds in an appropriate neighborhood of a point where the initial functions and the coefficients in \eqref{app:B1} are well-behaved. 

The general strategy behind the above is to replace the PDE by a system of ODEs which are considered easier to integrate. Yet, often, it is exactly the other way around that the PDE is easier to solve, for example by the method of separation of variables, than the system of ODEs. The idea of turning a system of ODEs into a PDE goes back to Hamilton and Jacobi and will be explained in a bit more detail in appendix \ref{app:C} in the context of classical mechanics.

\section{Hamilton-Jacobi~method and canonical~transformations}
\label{app:C}

This short summary of the Hamilton-Jacobi approach is mainly based on \arnold, more information can be found in any good mechanics book, e.g. in \cite{Goldstein}. To keep this small discourse as simple as possible, the discussion is restricted to time-independent, one-dimensional systems. Most of the formulas can be immediately generalized to models with more degrees of freedom. To also include time-dependent systems a bit more work would be required. 

\begin{definition}[Canonical transformation]
A map $g:\phasespace\to \phasespace$ is a canonical transformation iff
\bq
g^{\ast} \Omega=\Omega
\eq
where $g^{\ast}$ denotes the pull back and $\Omega$ the symplectic form on $\phasespace$.
\end{definition}
By applying Stokes theorem one can easily show that this condition is equivalent to 
\beq 
\label{appB1}
\oint_{\gamma} p\,\dif q- P\,\dif Q=0
\eeq
for any closed curve $\gamma$ in $\phasespace$. It is also well known that a canonical transformation leaves the equations of motion form-invariant, which can be proven by examining the transformation of the one-form $\Omega_1:=p\,\dif q -H\,\dif t$ on the extended phase space $\phasespace\times \R^+$. To see why, remember that any alternating two-form in odd-dimensions has at least one null-direction, i.e. it has at least one eigenvector to the eigenvalue zero. Here, this null-direction is given by the integral curve $(p(t),q(t))$ for which
\bq
0=\dif\Omega_1(p(t), q(t))=\dif p \wedge\dif  q-\frac{\partial{H}}{\partial p} \dif p\wedge \dif t-\frac{\partial{H}}{\partial q} \dif q\wedge \dif t
=(\frac{\dif q}{\dif t}-\frac{\partial{H}}{\partial p} \dif p\wedge \dif t)+(\frac{\dif q}{\dif t}+\frac{\partial{H}}{\partial q} )\dif q\wedge \dif t~.
\eq
This obviously implies the canonical equations of motion and consequently $\Omega_1$ captures the full dynamics.
\begin{theorem}[\arnold]
Suppose $g:\phasespace\to \phasespace$ is a canonical transformation that maps $(p,q)$ to $(P,Q)$ then there exist functions $K(P,Q)$ and $S(p,q)$ so that  
\begin{gather*}
p\,\dif q -H\,\dif t=P\,\dif Q -K\,\dif t +\dif S~,
\\[5pt]
\frac{\dif P}{\dif t}=-\frac{\partial K}{\partial Q} 
\quad\text{and}\quad
\frac{\dif Q}{\dif t}=\frac{\partial K}{\partial P}~.
\end{gather*}
\end{theorem}

\begin{proof}
Since condition \eqref{appB1} has to hold for all closed curves $\gamma$ the one-form $p\, \dif q-P\,\dif Q$ is exact, which means, that there exist a potential $S$ with $\dif S=p \dif q-P\dif Q$. Now set $K(P(p,q),Q(p,q))=H(p,q)$. This proofs the first part of the theorem. The second part follows directly from $\dif^2 S=0$.
 \end{proof}
 The function $S$ is called a \emph{generating function} of the canonical transformation $g$. For one-dimensional models there exist two types of such functions\footnote{For more than one degree of freedom there exist a third type depending on a mixture of $Q$ and $P$-variables.}:
\begin{itemize}
\item[A.] Suppose $\det\frac{\partial (Q,q)}{\partial (p,q)}\neq0$ then the momentum can be written as a function of $Q$ and $q$ by the inverse function theorem. Inserting $p=p(Q,q)$ in $S$ leads to $S(p,q)=S_1(Q,q)$. By comparison of $\dif S$ and $\dif S_1$ one finds 
\bq
\frac{\partial S_1(Q,q)}{\partial q}=p
\quad\text{and}\quad
\frac{\partial S_1(Q,q)}{\partial Q}=-P~.
\eq

\item[B.] If $\det\frac{\partial (P,q)}{\partial (p,q)}\neq0$  then $p=p(P,q)$. The corresponding generating function $S_2(P,q)$ is obtained via a Legendre transformation, that is, $S_2(P,q)=PQ+S(p,q)$. A comparison of the differentials yields
\bq
\frac{\partial S_2(P,q)}{\partial q}=p
\quad\text{and}\quad
\frac{\partial S_2(P,q)}{\partial P}=Q~.
\eq
\end{itemize}
A given function $S_1$/$S_2$ generates a canonical transformation iff $\frac{\partial^2 S_1}{\partial q\partial Q}\neq 0 $ / $\frac{\partial^2 S_2}{\partial q\partial P}\neq 0$. This non-degeneracy condition is needed to ensure that $Q$ / $P$ can be extracted as functions of $p$ and $q$. In more dimensions it must be replaced by $\det\frac{\partial^2 S_1}{ \partial (q,Q)}\neq 0$ or $\det\frac{\partial^2 S_2}{ \partial (q,P)}\neq 0$ respectively. 

The simplest example of a generating function is $S_2(P,q)=Pq$ that gives rise to the identity transformation. Another application is the Hamilton-Jacobi method. The main idea, hereby, is to transform the system such that the dynamics is especially simple. This is, of course, always the case if some coordinates are cyclic. Thus, one tries to find functions $S_1$ or $S_2$ so that  
\bq
H(\frac{\partial S}{\partial q}, q)=K(P,t)~.
\eq
Then $P$ is obviously constant and $Q(t)=\int_0^t \frac{\partial K}{\partial P}$ (for specific examples see section \ref{sec:examples}). 

This is closely related to the action-angle coordinates that can be introduced for models with compact level sets $M_{h}=\{(p,q)| H(p,q)=h\}$. Here, a function $S_2(I,q)$ generating the transformation $(p,q)\mapsto(I,\phi)$ is constructed that obeys
\bq
\frac{\partial S_2(I,q)}{\partial q}=p~,
\quad
\frac{\partial S_2(I,q)}{\partial I}=\phi
\quad\text{and}\quad
H(\frac{\partial S}{q}, q)=h(I)~.
\eq
For models with one degree of freedom $M_{h}$ being compact is equivalent with $M_{h}$ being a closed curve in $\phasespace$ that should be parametrized by $\phi$. This leads to the additional requirements 
\bq
I=I(h)\quad\text{and}\quad 
\oint_{M_h} \dif \phi=2\pi~.
\eq
As shown in \arnold, the function 
\bq
S_2(I,q)=\int_{\gamma_I(q_0,q)} p\,\dif q~,
\eq
where $\gamma_I(q_0,q)$ is a curve in $M_{h(I)}$ joining $q_0$ and $q$, meets all these requirements. 

%
%

\end{document}